\documentclass[twocolumn]{autart}    

\usepackage{graphicx}          
\usepackage{cite}
\usepackage[round, sort]{natbib}
\usepackage{booktabs}
\usepackage{amsmath,amssymb,amsfonts}
\usepackage{graphicx}
\usepackage{subfigure}

\usepackage{textcomp}
\usepackage{color}
\usepackage{enumerate}
\usepackage{makecell}
\usepackage{xcolor}
\usepackage{bm}
\usepackage{amstext}
\usepackage{stfloats}
\usepackage{subfigure}
\newtheorem{remark}{Remark}{}
\newtheorem{theorem}{Theorem}
\newtheorem{proposition}{Proposition}
\newtheorem{lemma}{Lemma}
\newtheorem{definition}{Definition}

\usepackage{mathrsfs}
\usepackage{multirow}
\usepackage{diagbox}
\usepackage{algorithm}
\usepackage{algpseudocode}
\allowdisplaybreaks[4]
\usepackage{cases}
\usepackage{subeqnarray}

\def\BibTeX{{\rm B\kern-.05em{\sc i\kern-.025em b}\kern-.08em
		T\kern-.1667em\lower.7ex\hbox{E}\kern-.125emX}}

\begin{document}

\begin{frontmatter}

\title{Variable-Horizon Model Predictive Control for Switched Systems}

\thanks[footnoteinfo]{This paper was not presented at any IFAC 
meeting. Corresponding author Yijing Wang.}
	\thanks{This work was supported by the National Natural Science Foundation of China (No. 62673349  and No. 62633015)  and the Hong Kong RGC--China NSFC Project N\_CityU133/2.}
	
\author[Paestum]{Rui Zhao}\ead{ruizhao@tju.edu.cn; ruzhao@cityu.edu.hk},    
\author[Rome]{Zhiqiang Zuo}\ead{zqzuo@tju.edu.cn},               
\author[Baiae]{Yang Shi}\ead{yshi@uvic.ca},  
\author[Rome]{Yijing Wang}\ead{yjwang@tju.edu.cn},
\author[Rome]{Zheng Li}\ead{zhengl@tju.edu.cn},
\author[Paestum]{Guanrong Chen}\ead{eegchen@cityu.edu.hk}

\address[Paestum]{Department of Electrical Engineering, City University of Hong Kong, Hong Kong SAR, P. R. China}  
\address[Rome]{Tianjin Key Laboratory of Intelligent Unmanned Swarm Technology and System, School of Electrical and Information Engineering, Tianjin University, Tianjin 300072, China}             
\address[Baiae]{Department of Mechanical Engineering, University of Victoria, Victoria, BC V8W 2Y2, Canada}        


\begin{abstract}                          
	This paper investigates model predictive control (MPC) for switched systems subject to control and state constraints. A variable-horizon switched MPC approach is proposed. By steering the system state into a well-designed switching feasible set, the proposed method structurally decouples the dwell-time conditions from the MPC constraints, thereby relaxing the dwell-time requirements to match those of the unconstrained switched systems.   Furthermore, algorithms are developed to construct this switching feasible set and characterize the domain of attraction, ensuring both persistent feasibility and closed-loop asymptotic stability.   To further decouple the prediction horizon length from strict dwell-time bounds, advanced short-horizon switched MPC schemes are designed, which expand the overall  domain of attraction. Simulations illustrate the efficacy of the proposed methods.
\end{abstract}

\end{frontmatter}

\section{Introduction}

A switched system consists of a family of subsystems and a logical rule that manages which subsystem is activated (\cite{DT}). In real-world applications, abrupt changes in structural parameters frequently occur because of intrinsic dynamics or external disturbances. Moreover, many systems operate in highly complex environments, exhibiting pronounced switching behaviors. Consequently, switched systems have emerged as a powerful modeling tool for a variety of applications such as bipedal robots (\cite{4343995}), aircraft engines (\cite{APP4}), smart grids (\cite{jaramillo2024dwell}), and so on.
It is well known that the stability of switched systems hinges on both the dynamics of the individual subsystems and the switching law. The switching signal (time-dependent, state-dependent, or hybrid) plays a critical role in determining whether the overall system is stable (\cite{acdef_GA}). Time-dependent switching strategies, such as dwell-time (\cite{hespanha1999stability}), average dwell-time (\cite{hespanha1999stability}), and mode-dependent average dwell-time (\cite{zhao2011stability}), are widely used to guarantee the system stability. 

Another challenge in control design is the presence of physical constraints  on control inputs and system states.  Model predictive control (MPC) addresses this challenge by predicting  the plant's behavior over a finite horizon and computing control actions that optimize a predefined performance index subject to these constraints (\cite{shi2021advanced}).
Since MPC relies on anticipated state evolution, its implementation generally requires accurate system models and switching signal information.

Most early MPC schemes for switched systems assume that the switching signal is \textcolor{black}{known \textit{a priori}} (\cite{mhaskarPredictiveControlSwitched2005}) or consider state-dependent switching  (\cite{abbasiRobustTubebasedMPC2021}). 
\textcolor{black}{Recently, infinite-horizon MPC laws have been co-designed with average dwell-time policies  utilizing linear matrix inequality techniques (\cite{nodoziLMIbasedModelPredictive2017}). Because the prediction horizon is inherently infinite, these approaches implicitly circumvent transient feasibility issues by leveraging  global Lyapunov-theoretic invariant sets. 
Stabilizing MPC without terminal constraints has also been investigated (\cite{knuferStabilizingModelPredictive2016}). Recent advances in this controlled switching paradigm seek to jointly optimize the control input and the switching signals under dwell-time constraints (\cite{Chen2022SwitchedMPC, zhuangModelPredictiveControl2023}). Although  effective, these active switching strategies typically require solving  computationally demanding mixed-integer optimization problems.}

In contrast, for strictly {exogenous} switching scenarios where the switching signal is unknown in advance and dictated by external environments, the controller cannot co-optimize the switching logic but must survive it. Guaranteeing the persistent feasibility of finite-horizon MPC under exogenous switching is notoriously challenging, as the pre- and post-switching models may differ drastically. For instance, while \cite{bridgemanStabilityFeasibilityMPC2016} developed a variable terminal set framework to ensure recursive feasibility, it strictly requires the switching sequence to be known in advance. When the switching is unknown, a key requirement is that the state at each switching instant must successfully transition into the feasible region of the subsequent mode. As demonstrated by \cite{zhangSwitchedMPCClass}, this creates a circular dependency: the dwell-time requirement depends on the switching-instant state, yet the latter must itself satisfy the constraints dictated by the dwell-time condition.

To break this circular dependency, the foundational work of \cite{zhang_switched_2016} established rigorous links among feasibility, attractivity, and stability via backward reachable set analysis. By enforcing rigid multi-step geometric intersections, the system state is compelled into the one-step reachable set at switching instants. While recent works of \cite{tanModelPredictiveControl2024} have extended this framework to broader contexts, the underlying geometric mechanism remains identical. The price of this fixed-horizon geometric approach is twofold: (i) the computational complexity grows with the system dimension and the number of modes, and (ii) the resulting dwell-time constraints are excessively restrictive.

Taking a different approach, a min-max switching MPC scheme developed by \cite{OngMPC} exploits a truncated admissible switching sequence. However, this min-max paradigm suffers from severe online computational demands and inherently yields a significantly restricted domain of attraction. \textcolor{black}{Alternatively, other recent strategies ensure recursive feasibility by constructing switch-robust control invariant (switch-RCI) sets  (\cite{Danielson2019_switchRCI}). \cite{Lavaei2020TubeMPC} further extended this idea to handle bounded disturbances under a tube-based framework. However, these studies focus on constraint enforcement rather than asymptotic convergence.}

%

\color{black}
Consequently, existing exogenous switching frameworks are highly difficult to extend to higher-dimensional or large-scale systems, owing to overly conservative offline dwell-time conditions, prohibitive online computational complexity, or the lack of rigorous stability guarantees.
To address the above-discussed challenging issues, this  paper proposes a \textit{variable-horizon} switched MPC strategy for switched systems subject to mode-dependent dwell-time constraints. The contributions of this paper are three-fold:\\
{\color{black}
	i) A novel \textit{variable-horizon} switched MPC framework is proposed, which structurally decouples the dwell-time conditions from the stringent MPC state and input constraints.  By actively steering the system state into a predefined switching feasible set, this framework aligns the dwell-time constraints with those required for unconstrained switched systems.
	Compared to the method of \cite{zhang_switched_2016}, which enforces strict dwell-time constraints to force persistent feasibility, our technique significantly reduces conservatism and provides structurally simpler conditions.\\	
	ii) Effective algorithms are developed to construct the switching feasible set and explicitly characterize the domain of attraction. Empowered by the design of the \textit{variable-horizon} mechanism and the switching feasible set, the proposed method guarantees persistent feasibility at switching instants alongside closed-loop asymptotic stability. Unlike the exhaustive min-max paradigm of \cite{OngMPC}, our approach requires a shorter prediction horizon and eliminates additional input restrictions, yielding a larger feasible region. Furthermore, it overcomes the theoretical limitations of \cite{Danielson2019_switchRCI} by  providing asymptotic stability guarantees.\\	
	iii) Short-horizon switched MPC schemes are used to decouple the required prediction horizon length from dwell-time. By tailoring the dwell-time constraint specifically for the initial mode transition, these schemes circumvent the geometric limitations of short horizons. Enlarging initialization dwell time expands the domain of attraction without sacrificing the relaxed dwell-time benefits for subsequent transitions.
}

Notation: For a matrix $A$, $A^\top$ denotes its transpose. The sets of real numbers and positive integers are denoted by $\mathbb{R}$ and $\mathbb{Z}_+$ respectively. The discrete interval of integers from $a$ to $b$ is represented by $\mathbb{Z}_{[a,b]}$. For a vector $x$ and positive definite matrix $P$, the weighted Euclidean norm is defined as $\|x\|_P \triangleq \sqrt{x^\top P x}$. Symbols $\succ$ and $\succeq$ imply positive definite and positive semi-definite. \textcolor{black}{A sequence of variables from index $a$ to $b$ is denoted by $\{x_i\}_{i=a}^{b} \triangleq \{x_a, x_{a+1}, \dots, x_b\}$.}

The remainder of this paper is organized as follows. Section \ref{sec2} gives some preliminaries and formulates the variable-horizon switched MPC problem. Section \ref{sec:feasibility} develops the algorithm for the switching feasible set and establishes the persistent feasibility of the proposed approach. Section \ref{sec:stability} computes the domain of attraction and proves the asymptotic stability of the system under the derived mode-dependent dwell-time constraints. To further mitigate the restrictive impact of prediction horizons on dwell-time constraints, Section \ref{sec:short} proposes short-horizon switched MPC schemes. Numerical simulations are presented in Section \ref{sec:sim} to support the obtained results. Finally, we conclude this paper in Section \ref{sec:col}. \textcolor{black}{To facilitate understanding, the overall framework and the logical relationships among these sections are illustrated by Fig. \ref{fig:para}.}


\begin{figure}
	\centering
	\includegraphics[width=1\linewidth]{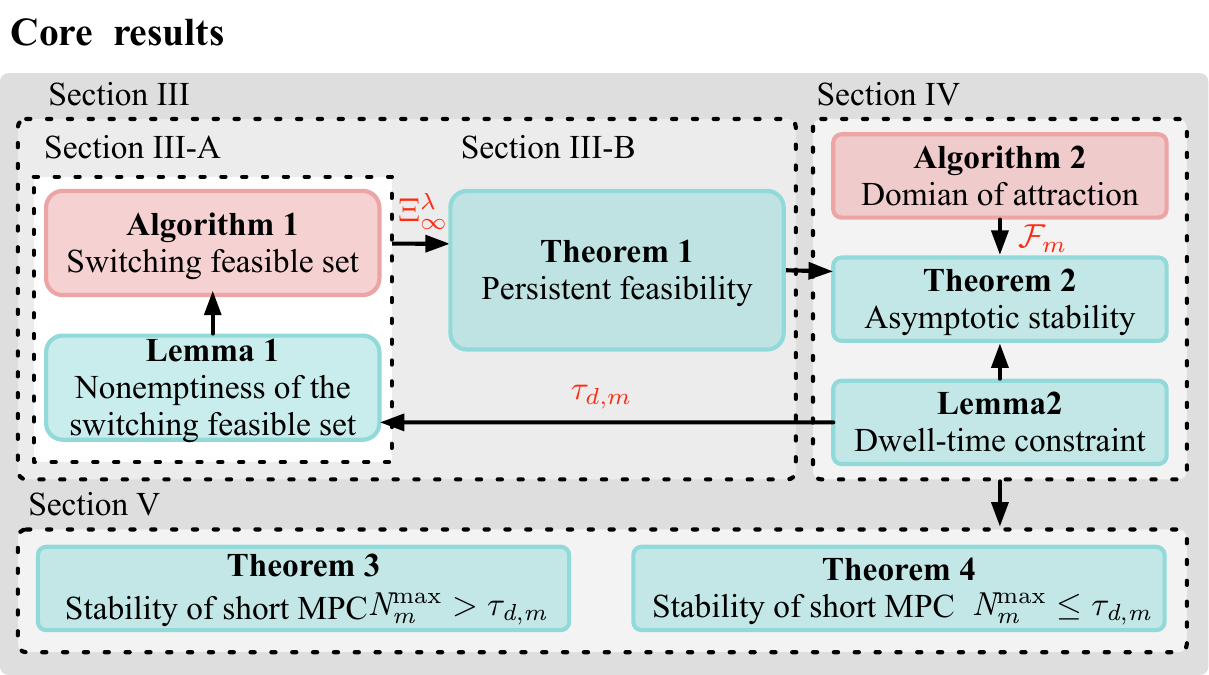}
	\caption{Schematic illustrating the relation of the presented algorithms and theories.}
	\label{fig:para}
\end{figure}

\section{Problem Formulation}\label{sec2}
Consider a class of discrete-time switched linear systems:
\begin{equation}\label{state}
	(\Omega_{\sigma(k)}): x_{k+1} = A_{\sigma(k)} x_k + B_{\sigma(k)} u_k,
\end{equation}
where $x_k \in \mathbb{R}^{n_x}$ denotes the system state and $u_k \in \mathbb{R}^{n_u}$ represents the control input. $\sigma(k) \in \mathcal{M} = \{1, \cdots, M\}$ is the switching signal with $M > 1$ being the number of subsystems. Here the switching signal $\sigma(k)$ is unknown in advance.  Assume that \((A_m,B_m)\) is stabilizable for every \(m\).
The state and input are subject to mode-dependent constraints: $x_k\in \mathbb{X}_m$ and $u_k \in \mathbb{U}_m$, where $\mathbb{X}_m \subseteq \mathbb{R}^{n_x}$ and $\mathbb{U}_m \subseteq \mathbb{R}^{n_u}$ are compact polyhedral sets containing the origin in their interiors. 


\begin{definition}[\cite{dehghan2013computations}]\label{def:dwell}
	For any two consecutive switching instants $k_{s}$ and $k_{s+1}$ with $s \in \mathbb{N}_{+}$, if $k_{s+1}-k_s \geq \tau_{d,m} > 0$ where $m = \sigma(k_s^+)\in \mathcal{M}$, then $\tau_{d,m}$ is called the mode-dependent dwell-time for the $m$-th subsystem. 
\end{definition}



\begin{figure*}[b]
	\centering
	\includegraphics[width=\linewidth]{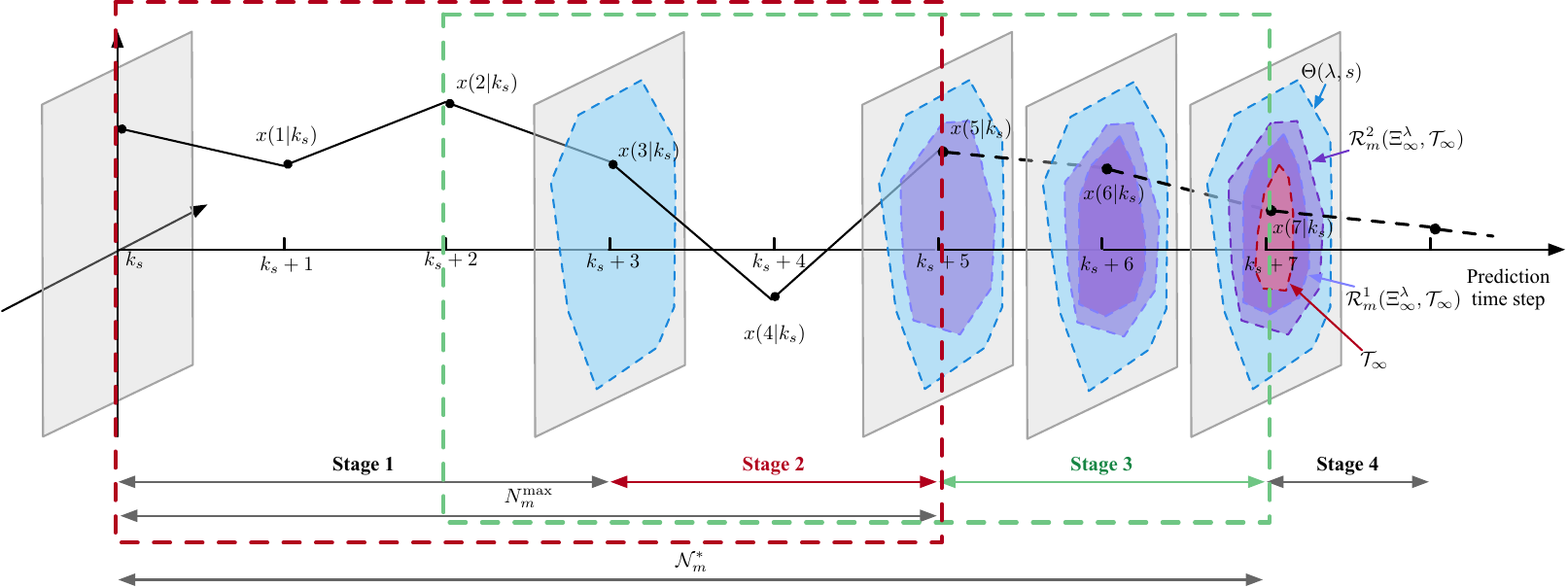}
	\caption{Prediction paradigm of the variable-horizon switched MPC at a switching instant $k_s$ (left) and $k_s+3$ (right) with $N^{\max}_m = 5$ and $N^{\min}_m= 2$.
		At $k_s$, the control sequence steers the state into the \textcolor{black}{contracted} switching feasible set $\Theta(\lambda,s)$ within $N_m^{\max}-N_m^{\min}$ steps, maintains it therein, and finally drives it into the terminal set $\mathcal{T}_\infty$. At $k_s+3$, the horizon reaches its minimum, reducing to a fixed-horizon MPC.
	}
	\label{fig:paradigm}
\end{figure*}

Let $\mathcal{D}(k) = k - k_s$ be the duration of the current mode $m = \sigma(k)$, where $k_s$ is the latest switching instant. The variable prediction horizon is defined as
\begin{equation}\label{equ_N}
	N_m(k,\mathcal{D}(k))= \max\left\{ N_m^{\max} - \mathcal{D}(k), N_m^{\min} \right\},
\end{equation}
where $N_m^{\min}$ and $N_m^{\max}$ are the minimum and maximum horizons, respectively. For brevity, let $\mathcal{N}_k^m \triangleq N_m(k,\mathcal{D}(k))$.  Define $\Delta N_m \triangleq N_m^{\max} - N_m^{\min}$ as the horizon range.

Cost function for $m \in \mathcal{M}$ is $	J_m(x_k, \bm{u}_k) \triangleq  T_m(x_{\mathcal{N}_k^m |k})+  \sum_{i=0}^{\mathcal{N}_k^m -1} \ell_m(x_{i|k}, u_{i|k})$. $\bm{u}_k \triangleq \{u_{0|k}, \ldots, u_{\mathcal{N}_k^m -1|k}\}$ denotes the predicted control sequence.  $u_{i|k}$ is the predicted control input at instant $k+i$. 
The symbol $x_{i|k}$ for $i \in \{1,\ldots,\mathcal{N}_k^m \}$ represents the state at instant $k+i$ under the control sequence $\bm{u}_k$. The cost functions are $	\ell_m(x_{i|k}, u_{i|k}) = \| x_{i|k} \|_{Q_m}^2 + \| u_{i|k} \|_{R_m}^2,$ $
T_m(x_{\mathcal{N}_k^m|k}) = \| x_{\mathcal{N}_k^m|k} \|_{P_m}^2,$
where $Q_m \in \mathbb{R}^{n_x \times n_x} \succ 0$, $R_m \in \mathbb{R}^{n_u \times n_u} \succ 0$ and $P_m \in \mathbb{R}^{n_x \times n_x} \succ 0$ are the state,  input  and terminal state weighting matrices, respectively.  \textcolor{black}{Here, to obtain the terminal-cost decrease used in the stability analysis, the terminal weighting matrix $P_m$ is chosen as the positive definite matrix solution to the DARE: $P_m = A_m^{\top} P_m A_m - A_m^{\top} P_m B_m (R_m + B_m^{\top} P_m B_m)^{-1} B_m^{\top} P_m A_m + Q_m$.  Let $K_{m}$ be the controller gain for the $m$-th subsystem in the form of $K_m = -(R_m + B_m^T P_m B_m)^{-1} B_m^T P_m A_m.$ Additionally, let $\overline{A}_m \triangleq A_m + B_m K_m$. }

The model predictive control optimization problem is formulated as follows:
\begin{subequations}\label{MPC}
	\begin{align}
		&\bm{u}_k^* = \arg\min_{\bm{u}_k} J_m(x_k, \bm{u}_k), \\
		\text{s.t.} ~ &x_{i+1|k} = A_m x_{i|k} + B_m u_{i|k}, ~ x_{0|k} = x_k, \label{MPC_c1} \\
		&x_{i|k} \in \mathbb{X}_m, ~ u_{i|k} \in \mathbb{U}_m, ~ \forall i \in \mathbb{Z}_{[0,~\mathcal{N}_k-1]}, \label{MPC_c2} \\
		&x_{i|k} \in \Theta(\lambda,s), ~ \forall i \in \mathbb{Z}_{[\mathcal{N}_k-N_m^{\min},~\mathcal{N}_k-1]}, \label{MPC_c3} \\
		&x_{\mathcal{N}_k|k} \in \mathcal{T}_{\infty}, \label{ter}
	\end{align}
\end{subequations}
where $\bm{u}_k^*$ stands for the optimal control sequence for the current mode $m = \sigma(k)$. $\Theta(\lambda,s)$ denotes the \textcolor{black}{contracted switching feasible set} parameterized by $\lambda$ with $s$ indicating the number of switches. $\mathcal{T}_{\infty}$ represents the terminal constraint set presented in \cite{Chong_2012}. Due to space limitations, this part is omitted here.

Fig. \ref{fig:paradigm} illustrates the multi-mode prediction paradigm for a second-order system at a switching instant $k_s$ and three steps later at $k_s+3$ (equivalent to a fixed-horizon MPC). The gray, blue, and red areas represent the state constraint $\mathbb{X}_{\sigma(k_s)}$, the  \textcolor{black}{contracted switching feasible set} $\Theta(\lambda,s)$, and the terminal set $\mathcal{T}_\infty$, respectively. Resetting the horizon to $N_{m}^{\max}$ at each switch guarantees that the state enters the  \textcolor{black}{contracted switching feasible set } $\Theta(\lambda, s)$ within $N_{m}^{\max} - N_{m}^{\min}$ steps and reaches $\mathcal{T}_\infty$ within $N_{m}^{\max}$ steps. The detailed step-by-step evolution of the variable horizon and the corresponding state constraints are summarized in Table \ref{table1}.

\begin{table}
	\centering
	\caption{Prediction horizon and state constraints at each instant and predicted step with $N^{\max}_m = 5$ and $N^{\min}_m = 2$ (see \eqref{MPC})}
	\label{table1}
	\renewcommand{\arraystretch}{1.2} 
	\begin{tabular}{c *{4}{c}}  
		\toprule
		& $k_s$ & $k_s+1$ & $k_s+2$ & $[k_s+3,k_{s+1})$ \\
		\midrule
		\textbf{Horizon} & 5 & 4 & 3 & 2 \\
		\textbf{Step 1} & $\mathbb{X}_m$ & $\mathbb{X}_m$ & $\Theta(\lambda,s)$ & $\Theta(\lambda,s)$ \\
		\textbf{Step 2} & $\mathbb{X}_m$ & $\Theta(\lambda,s)$ & $\Theta(\lambda,s)$ & $\mathcal{T}_\infty$ \\
		\textbf{Step 3} & $\Theta(\lambda,s)$ & $\Theta(\lambda,s)$ & $\mathcal{T}_\infty$ & N/A \\
		\textbf{Step 4} & $\Theta(\lambda,s)$ & $\mathcal{T}_\infty$ & N/A & N/A \\
		\textbf{Step 5} & $\mathcal{T}_\infty$ & N/A & N/A & N/A \\
		\bottomrule
	\end{tabular}\vspace{-6pt}
\end{table}

When $\mathcal{D}(k) = \Delta N_m$, the proposed MPC problem \eqref{MPC} reduces to a fixed-horizon MPC strategy
\begin{subequations}\label{fhMPC}
	\begin{align}
		\bm{u}_k^* &= \arg \min_{\bm{u}_k} J_m(x_k, \bm{u}_k), \\
		\text{s.t. } &x_{i+1|k} = A_m x_{i|k} + B_m u_{i|k}, ~x_{0|k} = x_k, \\
		&x_{i|k} \in \Theta(\lambda, s), ~u_{i|k} \in \mathbb{U}_m,~ \forall i \in \mathbb{Z}_{[0, N_m^{\min}-1]}, \\
		&x_{N_m^{\min}|k} \in \mathcal{T}_{\infty}.\vspace{-4pt}
	\end{align}
\end{subequations}
Here, the state constraints are simplified to $\Theta(\lambda, s)$, and the terminal state strictly enters $\mathcal{T}_{\infty}$ at step $N_m^{\min}$.
As illustrated by Stage 2 in Fig. \ref{fig:paradigm},  constraint \eqref{MPC_c3} represents the core mechanism of the variable-horizon scheme.
It is exactly this mechanism that guarantees persistent feasibility at future switching instants as will be rigorously analyzed in Section \ref{sec:feasibility}.

\vspace{-2pt}

The primary objectives of this paper are threefold: 1) To design the \textcolor{black}{contracted switching feasible set} $\Theta(\lambda, s)$; 2) To ensure persistent feasibility; 3) To compute the domain of attraction that guarantees stability for the \textit{variable-horizon} switched MPC problem \eqref{MPC}.

\vspace{-4pt}
\section{Persistent Feasibility}\label{sec:feasibility}

Persistent feasibility is a fundamental challenge in the MPC problem. 
For non-switched systems, it is guaranteed by driving the terminal state into a positive invariant set. However, unknown future switching signals render this standard approach invalid for switched systems, where feasibility is often vulnerable at the exact switching instants.
Motivated by the capability of variable-horizon MPC to steer states into a target set in finite time (\cite{variable_horizon}), this section rigorously addresses this challenge. 
By incorporating dwell-time constraints, a dedicated \textit{contracted switching feasible set} $\Theta(\lambda,s)$ is constructed , which effectively reformulates the feasibility requirement at switching instants and guarantees persistent feasibility for the proposed MPC problem \eqref{MPC}.

\subsection{Switching Feasible Set}\vspace{-4pt}
Let $\mathcal{P}_m(\mathcal{S}) \triangleq \{ x \in \mathbb{X}_m \mid \exists u \in \mathbb{U}_m, A_m x + B_m u \in \mathcal{S} \}$ denote the one-step backward reachable set (pre-set) for subsystem $m$. Its $l$-step recursive application is denoted by $\mathcal{P}_m^l(\mathcal{S})$, with $\mathcal{P}_m^0(\mathcal{S}) \triangleq \mathcal{S}$. Define  $\mathbb{T}_0^m \triangleq \{0, 1, \cdots, \tau_{d,m} - 1\}.$ 
\vspace{-5pt}

\begin{algorithm}
	\caption{Computation of  switching feasible set}\label{result_alg_sfs}
	\begin{algorithmic}[1]
		\Require ~~ $\lambda$, ${A}_m$, $B_m$, ${\mathbb{X}}_m$, $\mathbb{U}_m$, $\mathcal{T}_{\infty}$, $N^{\max}_m$ and $N^{\min}_m$
		\Ensure~~ ${\Xi}_{\infty}^{\lambda}$
		\State Set $k=0$ and let $\Xi_0^{\lambda} = \cap _{m\in \mathcal{M}}  {\mathbb{X}}_m  \cap _{t\in \mathbb{T}_0^m} \mathcal{P}^t_m({\mathbb{X}}_m) $\label{A2_step_1}
		\While{$\Xi_{k+1}^{\lambda} \neq \Xi_k^{\lambda}$}\label{A2_step_2}
		\State $\Xi_{k+1}^\lambda = \Xi_k^\lambda$
		\For{$m\in \mathcal{M}$}\label{A2_step_4}
		\State ${\mathcal{W}_0} = \mathcal{T}_\infty$\label{A2_step_5}
		\For{$t$ = $1$ to $N^{\min}_m$}\label{A2_step_6}
		\State ${\mathcal{W}} _t=    \mathcal{P}^1_m ( {\mathcal{W}_{t-1}})\cap  \lambda \Xi_{k}^{\lambda}$\label{A2_step_7}
		\EndFor \label{A2_step_8}
		\State		$\Xi_{k+1}^{\lambda} =  \Xi_{k+1}^\lambda \cap \mathcal{P}^{N^{\max}_m-N_{m}^{\min}}_m( {\mathcal{W}}_{N_m^{\min}})$ \label{A2_step_10}
		\EndFor 
		\State $k= k+1$
		\EndWhile
		\State$\Xi_{\infty}^{\lambda}= \Xi_{k+1}^{\lambda}$ 
	\end{algorithmic}
\end{algorithm}

\vspace{-5pt}

	\textit{Algorithm \ref{result_alg_sfs}} systematically constructs the switching feasible set, thereby formally guaranteeing the persistent feasibility of the proposed  MPC scheme. Specifically, Step \ref{A2_step_1} initializes the set to enforce state constraints up to the dwell-time boundary. The iterative backward reachability in Steps \ref{A2_step_4}--\ref{A2_step_10} rigorously ensures that for any $x_k \in \Xi_{\infty}^{\lambda}$, an admissible control sequence exists such that the predicted trajectory: (i) satisfies the common state constraints $\bigcap_{m \in \mathcal{M}} \mathbb{X}_m$ during the initial steps $i \in [1, N_m^{\max} - N_m^{\min} - 1]$; (ii) remains strictly within the contracted feasible set $\lambda \Xi_{\infty}^{\lambda}$ during the interval $i \in [N_m^{\max} - N_m^{\min}, N_m^{\max} - 1]$; and (iii) reaches the terminal set $\mathcal{T}_{\infty}$ exactly at step $N_m^{\max}$.

\begin{proposition}\label{result_properties}
The set $\Xi_{\infty}^\lambda$ generated by \textit{Algorithm \ref{result_alg_sfs}} satisfies: (i) ${\Xi}_{\infty}^\lambda \subseteq {\mathbb{X}}_m$ and $\Xi_{j+1}^\lambda \subseteq \Xi_j^\lambda$, $\forall j \geq 1$; (ii) $\Xi_{\infty}^\lambda$ is compact and contains the origin; (iii) $\Xi_\infty ^{\overline{\lambda}} \subseteq \Xi_\infty ^{{\lambda}} $ for $\overline{\lambda} \leq \lambda \leq 1$; and (iv) $\mathcal{T}_{\infty} \subseteq \Xi_{\infty}^\lambda$.
\end{proposition}

The proof follows directly from the iterative construction in \textit{Algorithm \ref{result_alg_sfs}} and is thus omitted.

\begin{lemma}\label{result_nonempty}
	If the unconstrained switched system \eqref{state} using controller gain $K_m$ is asymptotically stable under the mode-dependent dwell-time constraints $\{\tau_{d,m}\}$, then there exists a constant $\lambda \in (0,1)$ such that both the switching terminal set $\mathcal{T}_{\infty}$ and the switching feasible set $\Xi_{\infty}^\lambda$ are non-empty.  
\end{lemma}
\vspace{-0.4cm}
\begin{pf}
	By \cite[Theorem 6]{Chong_2012}, the stated asymptotic stability guarantees a non-empty set $\mathcal{T}_{\infty}$ for some $\lambda \in (0,1)$. It then follows from Proposition \ref{result_properties} (iv) that $\mathcal{T}_{\infty} \subseteq \Xi_{\infty}^\lambda$, which directly ensures $\Xi_{\infty}^\lambda \neq \emptyset$. \qed
\end{pf}
\vspace{-0.4cm}
\begin{remark}\label{remark_feasible_sw}
While $\Xi_{\infty}^\lambda$ is a subset of $\bigcap_{m\in \mathcal{M}} \mathcal{P}^{N_m^{\max}}_m(\mathcal{T}_{\infty})$, there is no explicit inclusion relationship between the contracted set $\Theta(\lambda,s)$ and the standard $N_m^{\min}$-step domain of attraction $\bigcap_{m\in \mathcal{M}} \mathcal{P}^{N_m^{\min}}_m(\mathcal{T}_{\infty})$. A detailed geometric analysis and visual demonstration of this property will be provided in Section \ref{sec:sim}.
\end{remark}

	{\color{black}
		\begin{remark}\label{remark_danielson_diff}
%
			The switching feasible set   $\Xi_{\infty}^\lambda$ differs from the mode-dependent switch-RCI framework of \cite{Danielson2019_switchRCI}. In \cite{Danielson2019_switchRCI},  each mode-dependent switch-RCI set is required to be one-step control invariant under the corresponding mode; thus, the state must remain in that set once it enters. This stepwise  containment severely restricts the allowable state space during offline computation. In contrast, our approach explicitly couples the set construction with the variable-horizon MPC mechanism. By resetting the prediction horizon at switching instants, the MPC provides  a transient buffer of $\Delta N_m$ steps. This co-design completely bypasses the conservative one-step invariance requirement, replacing it with a multi-step backward reachability condition that explicitly exploits the mode-dependent dwell time. Consequently, intermediate states may   temporarily evolve outside $\Xi_{\infty}^{\lambda}$ (provided that they remain within the physical bounds of $\mathbb{X}_m$). Furthermore, rather than merely guaranteeing the hard constraints in standard RCI frameworks, our scheme guarantees closed-loop asymptotic stability, as the terminal constraint ultimately drives the state into the terminal set $\mathcal{T}_{\infty}$.
		\end{remark}
}

\subsection{Persistent Feasibility }

This subsection proves the persistent feasibility of the proposed \textit{variable-horizon} switched MPC \eqref{MPC}.

\begin{theorem} \label{thm_fea}
	Suppose the MPC problem \eqref{MPC} is feasible at the initial instant, $\Delta N_{m} \leq \tau_{d,m}$ for all $m \in \mathcal{M}$, and $\Theta(\lambda, s)= \lambda^s { \Xi}_{\infty}^{\lambda}$  where $s \in \mathbb{Z}_{[0,\infty]}$ is the switch count. Then, \eqref{MPC} is persistently feasible for all future instants.
\end{theorem}
\vspace{-0.4cm}
\begin{pf}
	Let $\bm{x}^*_{k_0} = \{x^*_{i|k_0}\}_{i=1}^{\mathcal{N}^{m}_{k_0}}$ and $\bm{u}^*_{k_0}= \{u^*_{i|k_0}\}_{i=0}^{\mathcal{N}^{m}_{k_0}-1}$ be the optimal sequences at $k_0$ for mode $m = \sigma(k_0)$. At instant $k_0+1$, one has $x_{0|k_0+1} = x^*_{1|k_0}$ and $\mathcal{N}^{m}_{k_0+1} = \mathcal{N}^{m}_{k_0}-1$. Hence, the sequences $\hat{\bm{x}}_{k_0+1}= \{x^*_{i|k_0}\}_{i=2}^{\mathcal{N}^{m}_{k_0}}$ and $\hat{\bm{u}}_{k_0+1} = \{u^*_{i|k_0}\}_{i=1}^{\mathcal{N}^{m}_{k_0}-1}$ remain feasible. Iterating this process maintains feasibility over $[k_0, k_0 + \Delta N_m)$. For $k \in [k_0 + \Delta N_m, k_1)$, \eqref{MPC} reduces to the standard fixed-horizon MPC \eqref{fhMPC}. Since the terminal state enters the invariant set $\mathcal{T}_{\infty}$, recursive feasibility is inherently guaranteed by appending the local controller $K_m$.
At the first switch $k_1$, we have $x_{k_1} \in \Theta(\lambda, 0)$. \textit{Algorithm \ref{result_alg_sfs}} guarantees a feasible control sequence exists at this instant. By induction, the problem remains persistently feasible for all future intervals and switching instants. \qed
\end{pf}\vspace{-0.4cm}
\begin{remark}
	Theorem \ref{thm_fea} establishes a direct link between the variable horizon and dwell-time constraints. Since future switching signals are unknown, traditional methods like \cite{zhang_switched_2016} must guarantee feasibility under worst-case scenarios, relying on computationally intensive set operations to derive minimum dwell-times. Our approach reverses this paradigm: it utilizes predefined dwell-times to proactively design the horizon length. By actively steering states into predefined switching feasible set within the prescribed dwell time, the \textit{variable-horizon} mechanism bypasses these complex recursive calculations. This yields a twofold advantage: a significant reduction in computational complexity and a substantial improvement in real-world deployability.
\end{remark}

\section{Stability}\label{sec:stability}
 By Theorem \ref{thm_fea}, persistent feasibility hinges on a non-empty switching feasible set $\Xi_{\infty}^{\lambda}$ and adequate dwell-times.  This section first derives the dwell-time conditions ensuring $\Xi_{\infty}^{\lambda} \neq \emptyset$. We then characterize the domain of attraction and establish asymptotic stability.

Since $K_m$ stabilizes each subsystem, there exist $\rho_m > 0$ and $\lambda_m \in (0, 1)$ such that $\| (A_m + B_m K_m)^n \| \leq \rho_m \lambda_m^n$ for all $m \in \mathcal{M}$.

\begin{lemma}\label{result_lemme_dt}
	The unconstrained switched system \eqref{state} under controllers $u_k = K_{\sigma(k)} x_k$ is asymptotically stable if the dwell-times $\tau_{d,m}$ satisfy
	\begin{equation}\label{equ_dt}
		\rho_m \lambda_m^{\tau_{d,m}} < 1, \quad \forall m \in \mathcal{M}.
	\end{equation}
\end{lemma}
\vspace{-0.6cm}
\begin{pf}
	For any consecutive switches $k_s$ and $k_{s+1}$ with mode $m = \sigma(k_s)$, the state evolves as $x_{k_{s+1}} = (A_{m} + B_m K_m)^{k_{s+1}-k_s} x_{k_s}$.
Taking the norm and applying $k_{s+1} - k_s \geq \tau_{d,m}$ yields $\| x_{k_{s+1}} \| \leq \rho_m \lambda_m^{\tau_{d,m}} \| x_{k_s} \|$. Condition \eqref{equ_dt} ensures a strict contraction $\| x_{k_{s+1}} \| < \| x_{k_s} \|$ at each switch. By recursion, $\lim_{k \to \infty} \|x_k\| = 0$, establishing asymptotic stability. \qed
\end{pf}
\vspace{-0.4cm}
\begin{algorithm}
	\caption{Computation of the Domain of Attraction}\label{result_alg_fs}
	\begin{algorithmic}[1]
		\Require ~~ ${A}_m$, $B_m$, ${\mathbb{X}}_m$, $\mathbb{U}_m$, ${\Xi}_{\infty}^{\lambda}$, $\mathcal{T}_{\infty}$, $m$, $N^{\max}_m$, $N^{\min}_m$
		\Ensure~~ ${\mathcal{F}}_{m}$
		\State  ${\mathcal{W}_0^m} = \mathcal{T}_\infty$
		\For{$t$ = $1$ to $N^{\min}_m$} \label{A3_step2}
		\State  ${\mathcal{W}_{t}^m}  = {\mathcal{W}}_{t-1}^m \cup \left(\mathcal{P}^1_m ( {\mathcal{W}_{t-1}^m}) \cap  \Xi_{\infty}^{\lambda}\right)$ \label{A3_step3}
		\EndFor \label{A3_step4}
		\State  ${\mathcal{F}}_{m} = {\mathcal{W}}_{N_m^{\min}}^m \cup \mathcal{P}^{N^{\max}_m-N^{\min}_m}_m (  {\mathcal{W}}_{N_m^{\min}}^m)$ \label{A3_step_5}
	\end{algorithmic} 
\end{algorithm}

\vspace{-4pt}
\textit{Algorithm \ref{result_alg_fs}} outlines the computation of the domain of attraction $\mathcal{F}_{m}$ for \eqref{MPC}.
Problem \eqref{MPC} is initially feasible for mode $m$ if $x_0 \in \mathcal{F}_m$.  
Notably,  the set $\Xi_{\infty}^{\lambda}$ constitutes a subset of each mode-specific feasible set, i.e., $\Xi_{\infty}^{\lambda} \subseteq \bigcap_{m \in \mathcal{M}} \mathcal{F}_{m}.$


\begin{theorem}\label{result_thm_sta}
	Consider the MPC problem (\ref{MPC}) for the switched system (\ref{state}) with horizon satisfying (\ref{equ_N}). Suppose $\Delta N_{m} \leq \tau_{d,m}$ for all $m \in \mathcal{M}$, the dwell-times $\tau_{d,m}$ satisfy \eqref{equ_dt}, and the initial state $x_0 \in \mathcal{F}_{\sigma(0)}$ is computed via \textit{Algorithm \ref{result_alg_fs}}. Then, the closed-loop switched system under \eqref{MPC} is asymptotically stable.
\end{theorem}

\vspace{-0.4cm}
\begin{pf}
	\textbf{Feasibility:} Lemmas \ref{result_nonempty} and \ref{result_lemme_dt} guarantee that
\(\mathcal T_\infty\) and \(\Xi_\infty^\lambda\) are nonempty for some \(\lambda\in(0,1)\). Since \(x_0\in\mathcal F_{\sigma(0)}\), problem \eqref{MPC} is initially feasible, and its persistent feasibility follows from Theorem 1.	


\textbf{Stability:} Let \(V_k=J_{\sigma(k)}^*(x_k)\) denote the optimal value of problem (3). Consider two consecutive instants \(k\) and \(k+1\) within the same mode \(m\). During the decreasing-horizon phase, removing the first element of the optimal sequence at \(k\) gives a feasible candidate at \(k+1\). Once the horizon reaches \(N_m^{\min}\), a feasible candidate is obtained by shifting the optimal sequence and appending the terminal control \(K_mx\). From the DARE identity,
\(
T_m((A_m+B_mK_m)x)-T_m(x)
=-\ell_m(x,K_mx).\)
Consequently,  $	V_{k+1}-V_k
\leq-\ell_m(x_k,u_k^*).$

Let \(k_s\) denote the switching instants. By constraint (3d), Algorithm 1, and \(\Delta N_m\leq\tau_{d,m}\), \(x_{k_{s+1}}\in
\Theta(\lambda,s)
=\lambda^s\Xi_\infty^\lambda.\)
Because the dynamics are linear, the cost is quadratic, and the feasible transition constructed by Algorithm 1 can be scaled by \(\lambda^s\), there exists a constant \(c_V>0\), independent of \(s\), such that $V_{k_{s+1}}
\leq c_V\lambda^{2s}.$

If infinitely many switches occur, then for every \(k\in[k_{s+1},k_{s+2})\), we have $V_k\leq c_V\lambda^{2s}.$
Moreover, with $\underline q
=\min_{m\in\mathcal M}\lambda_{\min}(Q_m)>0,$
one has $\underline q\|x_k\|^2
\leq V_k
\leq c_V\lambda^{2s}.$
Since \(\lambda\in(0,1)\), this implies \(x_k\to0\).

If only finitely many switches occur, the system eventually remains in one fixed mode. After the final switch, one has $\sum_{k=k_s}^{\infty}
\ell_m(x_k,u_k^*)
\leq V_{k_s}<\infty.$
Since \(Q_m\succ0\), it follows that \(x_k\to0\). Therefore, the closed-loop system is asymptotically stable. \(\square\)

\end{pf}
\vspace{-0.4cm}
\color{black}
\begin{remark}\label{remark2}
	Compared with \cite{zhang_switched_2016}, the proposed dwell-time condition has a simpler form and yields less conservative bounds. Through the co-design of the variable prediction horizon and the switching feasible set, feasibility under the hard state and input constraints is handled by the horizon and set conditions. Once the stabilizing feedback gain \(K_m\) is obtained from the DARE for the chosen LQR weighting matrices, the transient factor \(\rho_m\) and decay rate \(\lambda_m\) are determined from \(|(A_m+B_mK_m)^n|\leq\rho_m\lambda_m^n\). Condition of $\rho_m\lambda_m^{\tau_d,m}<1$ remains  sufficient whose conservatism depends on the selected feedback gain, norm, and exponential estimate, but it is not further restricted  by explicit hard-constraint or cross-mode set calculations.
\end{remark}

{\color{black}	
	\begin{remark}
%
The proposed framework can be extended to a robust tube-based MPC formulation for switched systems subject to bounded additive disturbances. 
Following \cite{MAYNE2005219}, for each mode $m$, a local gain $K_ m$ and a robust positively invariant set $\mathbb Z_m$ can be constructed for the error dynamics, yielding the tightened nominal constraints of $\hat{\mathbb{X}}_m = \mathbb{X}_m \ominus \mathbb{Z}_m$ and $\hat{\mathbb{U}}_m = \mathbb{U}_m \ominus K_m \mathbb{Z}_m$. 
The predecessor-set computations underlying Algorithms 1 and 2 could then be adapted to these tightened constraint sets to seek corresponding nominal terminal and switching feasible sets.
Let $\hat\Theta_m(\lambda,s)$ denote a robust nominal switching target set to be designed. Its construction should jointly address constraint tightening, inter-step reachability or invariance, tube re-centering under admissible mode transitions, and an appropriate practical stability condition. Since the mode-dependent tube does not contract with $\lambda$, nominal-set scaling alone could not establish contraction of the actual-state set or closed-loop stability. The detailed set construction and practical-stability analysis are left for future investigation.
	\end{remark}
}
{\color{black}
	\begin{remark}
The variable-horizon mechanism introduces a trade-off between closed-loop performance and switching feasibility. Shortening the prediction horizon reduces the prediction capability and may therefore degrade  performance. However, this effect can be mitigated because the full prediction horizon is used immediately after each switch and shortens only as the state is driven toward the switching feasible set. The DARE-based terminal weight \(P_m\) further mitigates the truncation effect by accounting for the cost beyond the shortened horizon; it represents the exact infinite-horizon tail cost for the corresponding unconstrained fixed-mode LQR problem.
%
%
\end{remark}}
\color{black}


	The horizon parameters distinctly shape the switching feasible set $\Xi_{\infty}^{\lambda} \subseteq \bigcap_{m \in \mathcal{M}} \mathcal{P}^{N_m^{\max}}_m(\mathcal{T}_{\infty})$. 
	Increasing $N_m^{\max}$ while fixing $N_m^{\min}$ enlarges the span $\Delta N_m$, thereby expanding the feasible set. Conversely, increasing $N_m^{\min}$ with a fixed $N_m^{\max}$ narrows $\Delta N_m$, creating a trade-off where stricter transition conditions may paradoxically shrink the domain of attraction. 
	Furthermore, synchronously shifting the horizon window, i.e., increasing both $N_m^{\max}$ and $N_m^{\min}$ while keeping $\Delta N_m$ constant, initially expands the switching feasible set. However, this growth saturates beyond a critical threshold, as the feasible space is ultimately bounded by the system's inherent maximal stabilizable set. 
More importantly, blindly pursuing a larger domain of attraction by increasing $N_m^{\max}$ inevitably exacerbates the computational burden and demands longer mode-dependent dwell-time (since $\Delta N_m \le \tau_{d,m}$). To break this geometric and computational bottleneck, the following section introduces a \textit{short-horizon} switched MPC scheme, designed to effectively enlarge the feasible region without relying on excessively long prediction horizons.

\section{Short-Horizon Switched MPC}\label{sec:short}
The feasible set is determined by the horizon bounds $N_{m}^{\max}$, $N_m^{\min}$, and the switching feasible set. Fixing the prediction horizon at $N_{m}^{\max}$ as in \cite{zhang_switched_2016} shrinks $\Xi_{\infty}^{\lambda}$. Consequently, the overall feasible set of the proposed variable horizon MPC is smaller than that in \cite{zhang_switched_2016}.  In addition, the \textit{variable horizon} MPC requires $N_m^{\max} \geq \tau_{d,m}$ and an upper bound of $N_m^{\max} \leq \tau_{d,m} + N_m^{\min}$. If the predefined dwell time is close to the maximum horizon limit, the required minimum horizon $N_m^{\min}$ cannot be sufficiently small, which severely restricts the feasible set.
To mitigate this restriction, we propose a \textit{short horizon} switched MPC scheme for two cases: (i) $N_m^{\max} \geq \tau_{d,m}$ and (ii) $N_m^{\max} < \tau_{d,m}$.

 Let $\mathcal{N}_{m}^{*}$ and $N^{\max}_{m}$ denote the desired and actual maximum prediction horizons, respectively. Define $\underline{\mathcal{N}}^*_m$ as the desired minimum prediction horizon within the switching feasible set. The relationship between $\underline{\mathcal{N}}^*_m$ and $N_m^{\min}$ is analogous to that between ${\mathcal{N}}^*_m$ and $N_m^{\max}$.  The state enters the terminal set $\mathcal{T}_\infty$ within ${\mathcal{N}}^*_m$ steps, and it reaches and stays in the switching feasible set in exactly $\underline{\mathcal{N}}^*_m$ steps after a switch.
 For set computations, define a generalized predecessor operator $\mathcal{R}_m(\mathbb{A},\mathcal{S}) \triangleq \{x \in \mathbb{A} \mid \exists u \in \mathbb{U}_m, A_m x + B_m u \in \mathcal{S}\}$. Its recursive application is denoted by $\mathcal{R}^l_m(\mathbb{A},\mathcal{S}) = \mathcal{R}(\mathcal{R}^{l-1}_m(\mathbb{A},\mathcal{S}))$ for $l \ge 1$, initialized with $\mathcal{R}_m^0(\mathbb{A},\mathcal{S}) \triangleq \mathcal{S}$.

\subsection{Short-Horizon Switched MPC for {$\mathit{N^{\max}_m > \tau_{d,m}}$}}\label{sec:sh}

For the case of $N^{\max}_m > \tau_{d,m}$, one can directly select a minimum horizon $N_m^{\min} \geq N_m^{\max} - \tau_{d,m}$.
The short-horizon switched MPC scheme is formulated as follows:
\begin{subequations}\label{SHMPC1}
	\begin{align}
		& \bm{u}_k^* = \arg \min_{\bm{u}_k} J_m (x_k, \bm{u}_k) \\
		\text{s.t.} & \quad x_{i+1|k} = A_m x_{i|k} + B_m u_{i|k}, ~ x_{0|k} = x_k, \label{SHMPC_c1} \\
		& \quad x_{i|k} \in \mathbb{X}_m, ~ u_{i|k} \in \mathbb{U}_m, ~ \forall i \in \mathbb{Z}_{[0,~\mathcal{N}^m_k-1]}, \label{SHMPC_c2} \\
		& \quad x_{i|k} \in \Theta(\lambda, s), \quad \forall i \in \mathbb{Z}_{[\delta_k, ~\mathcal{N}^m_k]}, \label{SHMPC_c3} \\
		& \quad x_{\mathcal{N}^m_k|k} \in \mathcal{R}_m^{\delta_k}(\Xi_{\infty}^{\lambda},\mathcal{T}_{\infty}), \label{SHter}
	\end{align}
\end{subequations}
where $	\delta_k = \max \{\Delta \mathcal{N}_m^* - \mathcal{D}(k), 0\},$ $\Delta \mathcal{N}_m^*= \mathcal{N}_{m}^{*} - N_{m}^{\max}$ 
and
$	\mathcal{N}^m_k 
	= \begin{cases}
		N_{m}^{\max},& \mathcal{D}(k) \leq  \Delta \mathcal{N}_m^*, \\
		\max \{	\mathcal{N}_{m}^{*}-\mathcal{D}(k),N_{m}^{\min}\},&  \mathcal{D}(k) >  	\Delta \mathcal{N}_m^* .
	\end{cases}
$

Ideally, the state converges into the terminal set $\mathcal{T}_\infty$ within $\mathcal{N}^*_m$ steps. However, since the allowable maximum prediction horizon is strictly bounded by $N_m^{\max} < \mathcal{N}_m^*$, the controller cannot sense the terminal set directly at the switching instant. To bridge this temporal gap and ensure the state systematically reaches $\mathcal{T}_\infty$, intermediate predecessor sets are introduced. Once the mode duration becomes sufficiently long, this short horizon scheme smoothly reduces to the standard variable horizon MPC framework. \textcolor{black}{For detailed schematic diagrams and step-by-step constraint evolution tables, see Appendix.}


\begin{theorem}\label{result_thm_SHMPC}
	Consider the \textit{short-horizon} switched MPC problem (\ref{SHMPC1}) for the switched system (\ref{state}). If the horizon parameters satisfy $\mathcal{N}_m^* - N_m^{\min} \leq \tau_{d,m} < N_m^{\max},$ the dwell-time $\tau_{d,m}$ fulfills (\ref{equ_dt}), and $x_0 \in \mathcal{F}_{\sigma(0)}$, where  $\mathcal{F}_m$ is constructed by \textit{Algorithm \ref{result_alg_fs}} using the  prediction horizon $\mathcal{N}_m^*$, then the switched system (\ref{state}) is asymptotically stable.
\end{theorem}
\textcolor{black}{The proof is omitted for brevity and is provided in Appendix.}


\subsection{Short-Horizon Switched MPC for $\mathit{{N_m^{\max} \leq  \tau_{d,m}}}$}
Consider the case where $N_m^{\max} \leq \tau_{d,m}$. Under this condition, setting $\underline{\mathcal{N}}^*_m = N_m^{\min}$ is no longer feasible; thus, a modified MPC formulation is required. For any state constraint $x_{i|k} \in \mathcal{S}$ over $i \in \mathbb{Z}_{[a,b]}$, the constraint is void if $a > b$.
Define $
{\mathcal{R}}_m^{i}(\mathbb{X}_m,\Theta(\lambda,s)) = \mathcal{R}_m^{i}(\mathbb{X}_m, \Theta(\lambda,s)).
$

The short-horizon MPC problem is formulated as follows:
\begin{subequations}\label{SHMPC2}
	\begin{align}
		& \bm{u}_k^* = \arg \min_{\bm{u}_k} J_m(x_k, \bm{u}_k) \\
		\text{s.t. ~} & x_{i+1|k} = A_m x_{i|k} + B_m u_{i|k}, ~ x_{0|k} = x_k, \\
		& x_{i|k} \in \mathbb{X}_m, \quad u_{i|k} \in \mathbb{U}_m, ~ \forall i \in \mathbb{Z}_{[0, N_m^{\max}-1]}, \\
		& x_{i|k} \in \Theta(\lambda, l), ~ \forall i \in \mathbb{Z}_{[\varepsilon_k, \mathcal{N}^m_k-1]}, \label{equ_SHMPC2} \\
		& x_{N_m^{\max}|k} \in {\mathcal{R}}^{\rho_k}_m(\mathbb{X}_m, \mathcal{R}_m^{\delta_k}(\Xi_{\infty}^{\lambda},\mathcal{T}_{\infty}) ),
	\end{align}
\end{subequations}
where $\varepsilon_k= \max\{\mathcal{N}_{m}^{*}  - \underline{\mathcal{N}}_m^*- \mathcal{D}(k), 0\}$, $	\rho_k = \max\{ \mathcal{N}_m^*  - \underline{\mathcal{N}}_m^*+N_m^{\max}- \mathcal{D}(k),0\}, $ $	\delta_k	=~\begin{cases}
	\varepsilon_k& \rho_k=0, \\
	\mathcal{N}_{m}^{*}  - \underline{\mathcal{N}}_m^* &  \rho_k>0 ,
\end{cases}$
and the horizon length $\mathcal{N}^m_k$ is determined by
\begin{equation*}
	\begin{aligned}
		&\mathcal{N}^m_k 
		= \begin{cases}
			N_{m}^{\max},& \mathcal{D}(k) \leq  \Delta	\mathcal{N}_{m}^{*},\\
			\max \{	\mathcal{N}_{m}^{*}-\mathcal{D}(k),N_{m}^{\min}\},&  \mathcal{D}(k) >  \Delta	\mathcal{N}_{m}^{*}  .
		\end{cases}
	\end{aligned}
\end{equation*}
 Note that, if $\delta(k, \mathcal{D}(k)) > \mathcal{N}^m_k$, the state constraint \eqref{equ_SHMPC2} becomes inactive. For a detailed step-by-step illustration of the constraint evolution and the specific design of the terminal sets, one is referred to Table \ref{table4} in Appendix.
 
 \begin{theorem}\label{result_thm_SHMPC2}
 	Consider the short-horizon switched MPC problem \eqref{SHMPC2}. Suppose that the dwell time $\tau_{d,m}$ satisfies \eqref{equ_dt}, the horizon parameters satisfy $\mathcal{N}_m^* - \underline{\mathcal{N}}_m^* \leq \tau_{d,m}$, and the initial state $x_0 \in \mathcal{F}_{\sigma(0)}$. If the domain of attraction $\mathcal{F}_m$ is computed by Algorithm \ref{result_alg_fs} using the maximum horizon $\mathcal{N}_m^*$, then the closed-loop system is asymptotically stable. 
 \end{theorem}
 
 The proof parallels that of Theorem \ref{result_thm_SHMPC} and is omitted for brevity. Evidently, problem \eqref{SHMPC2} relaxes the coupling between the dwell time and the prediction horizon bounds.

\begin{remark}\label{remark_short}
	The proposed MPC formulations \eqref{MPC}  typically yield a smaller domain of attraction compared to (\cite{zhang_switched_2016}). To overcome this, the proposed short-horizon scheme strategically enlarges the initial dwell time $\hat{\tau}_{d,m}$ relative to the values of $\tau_{d,m}$ used in Algorithm \ref{result_alg_sfs}. For fixed $N_m^{\min}$ and $N_m^{\max}$, the desired maximum horizon is defined as $\mathcal{N}^*_m = N_m^{\min} + \hat{\tau}_{d,m}$. Depending on the relationship between $N_m^{\max}$ and $\hat{\tau}_{d,m}$, the control sequence is computed via \eqref{SHMPC1} or \eqref{SHMPC2}. The resulting domain of attraction coincides exactly with $\mathcal{F}_m$ generated by Algorithm \ref{result_alg_fs} under $\mathcal{N}^*_m$. This confirms that the proposed method systematically expands the feasible region, validated by the subsequent simulations. 
\end{remark}

\section{Simulations }\label{sec:sim}
Consider a switched system consisting of two subsystems with parameters
{\scriptsize $A_{1} = \left[ \begin{array}{cc}	0.1397 & -0.9598 \\	1.4822 & 0.3995	\end{array}\right], $ $ 	B_{1} = \left[ \begin{array}{c}	0.5 \\ 	1\end{array}\right], $ $A_{2} = \left[ \begin{array}{cc}
		0.4557 & -1.0797 \\
		1.1777 & -0.4360
	\end{array}\right], $} {\scriptsize $ 	B_{2} = \left[ \begin{array}{c}	0 \\ 	1	\end{array}\right].$}
The state and control constraints are given by $\mathbb{X}_1 = \mathbb{X}_2 = \{ x \in \mathbb{R}^2 : \|x\|_{\infty} \leq 5\}$, 
$\mathbb{U}_1 = \{ u \in \mathbb{R} : \|u\| \leq 1.5\}$ and $ \mathbb{U}_2 = \{ u \in \mathbb{R} : \|u\| \leq 1\}.$ 
The state weighting matrices are $Q_1 = Q_2 = I$, and the input weighting matrices are $R_1 = R_2 = 1$. The terminal penalty matrices $P_m$ and feedback gains $K_m$ ($m \in \{1,2\}$) are derived using the unconstrained LQR method. According to Lemma \ref{result_lemme_dt}, the dwell-time constraints  are 
$\tau_{d,1}= \tau_{d,2}= 2$. Set  the horizons to $N^{\min}_1 = N^{\min}_2 = 5$, and $N^{\max}_1 =N^{\max}_2 = 7$.


\begin{figure}
	\centering
	\subfigure[Maximum prediction horizon]{\includegraphics[width=0.65\linewidth]{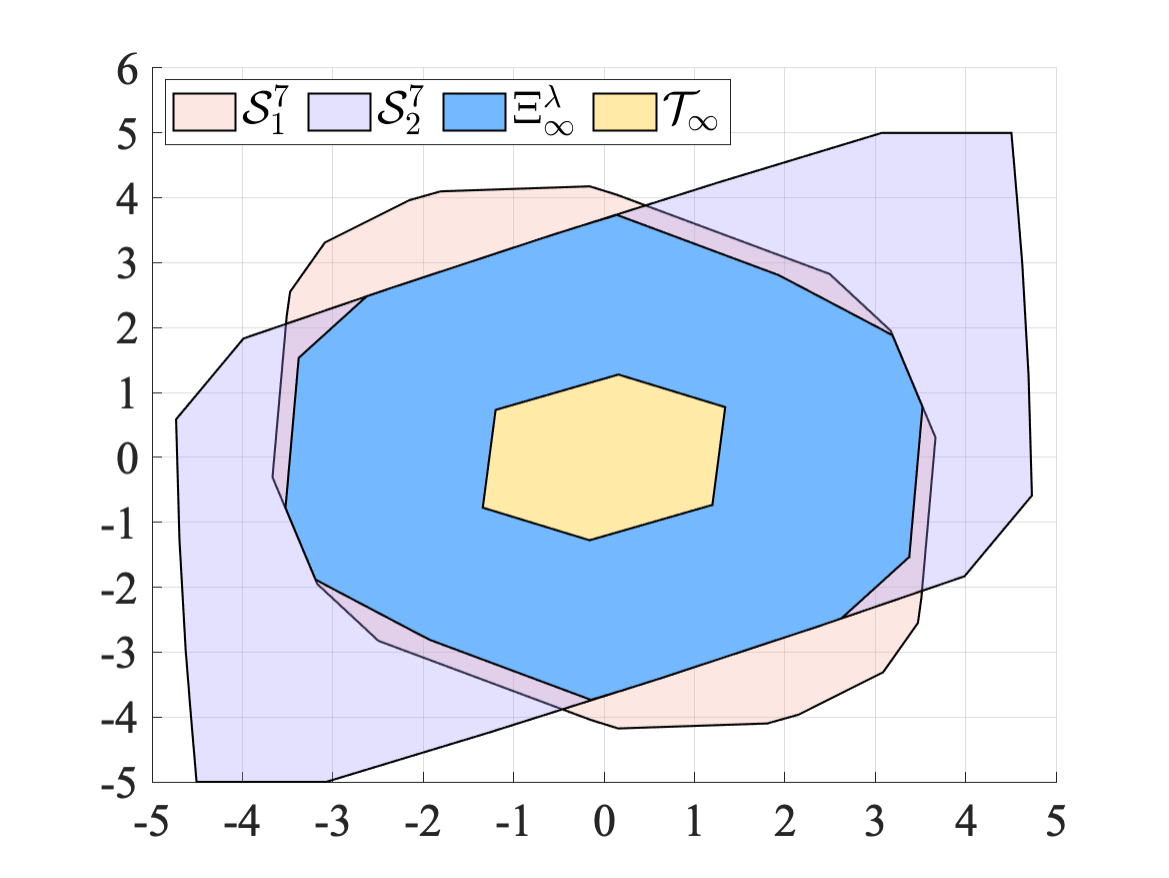}}\vspace{-0.2cm}
	\subfigure[Minimum prediction horizon]{\includegraphics[width=0.65\linewidth]{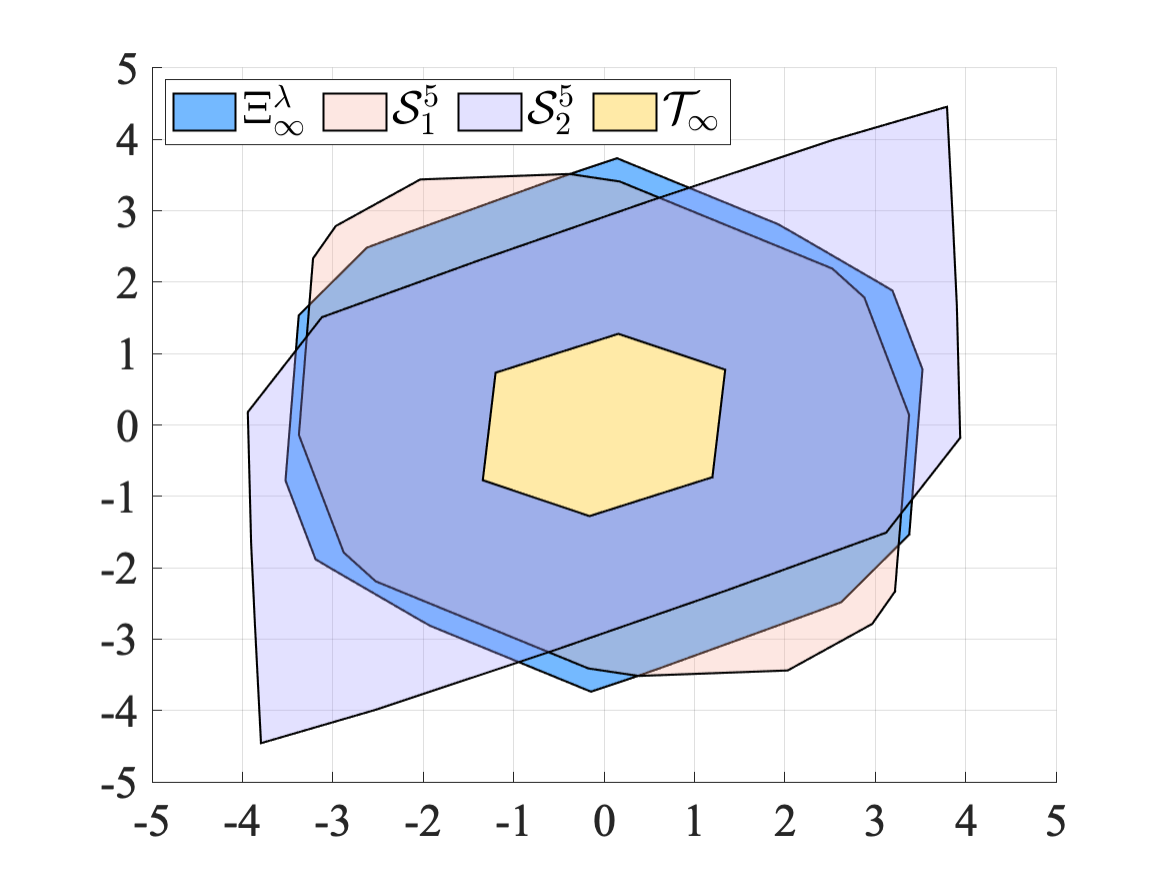}}\vspace{-0.3cm}
	\caption{The feasible regions with $N_1^{\max} = N_2^{\max} = 7$ and $N_1^{\min} = N_2^{\min} = 5$ for the standard MPC and the corresponding switching feasible set} \vspace{-0.4cm}
	\label{fig:set}
\end{figure}

\begin{figure}
	\centering
	\subfigure[Maximum prediction horizon]{\includegraphics[width=0.65\linewidth]{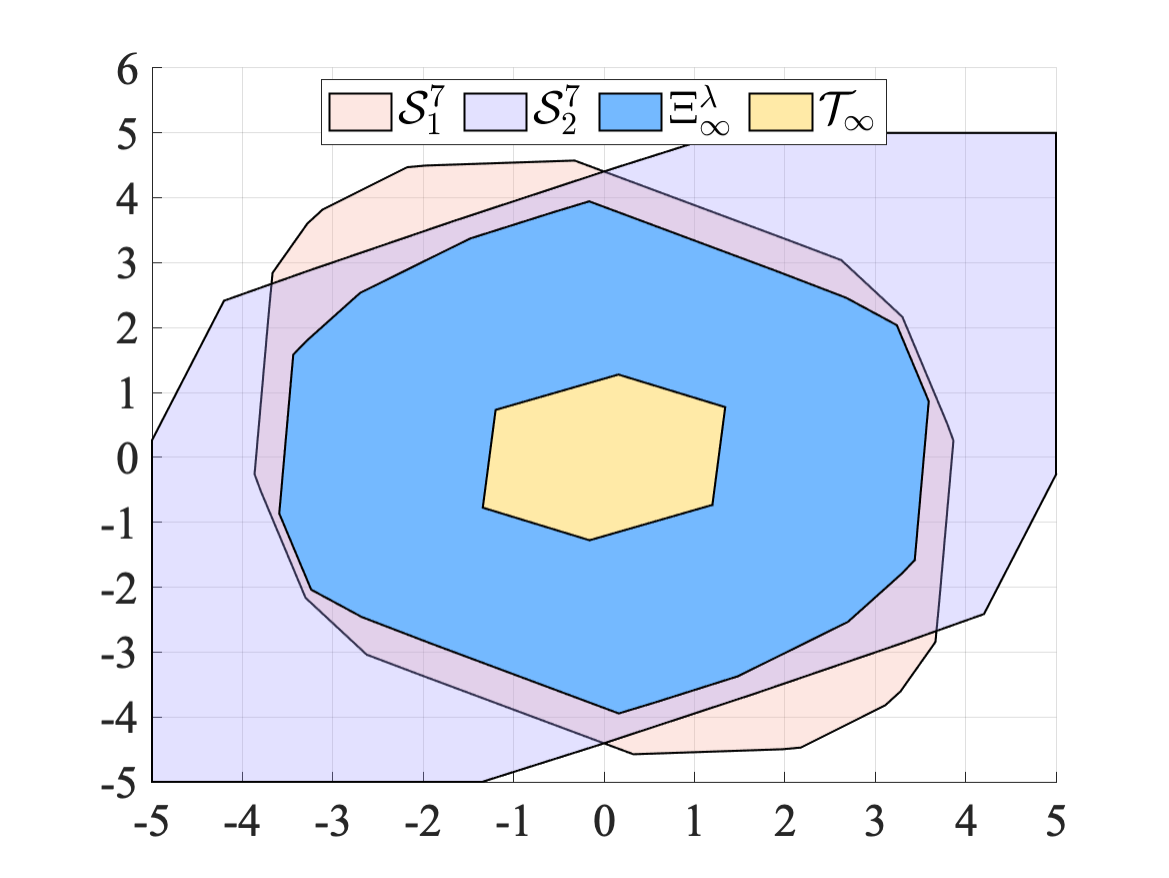}}\vspace{-0.2cm}
	\subfigure[Minimum prediction horizon]{\includegraphics[width=0.65\linewidth]{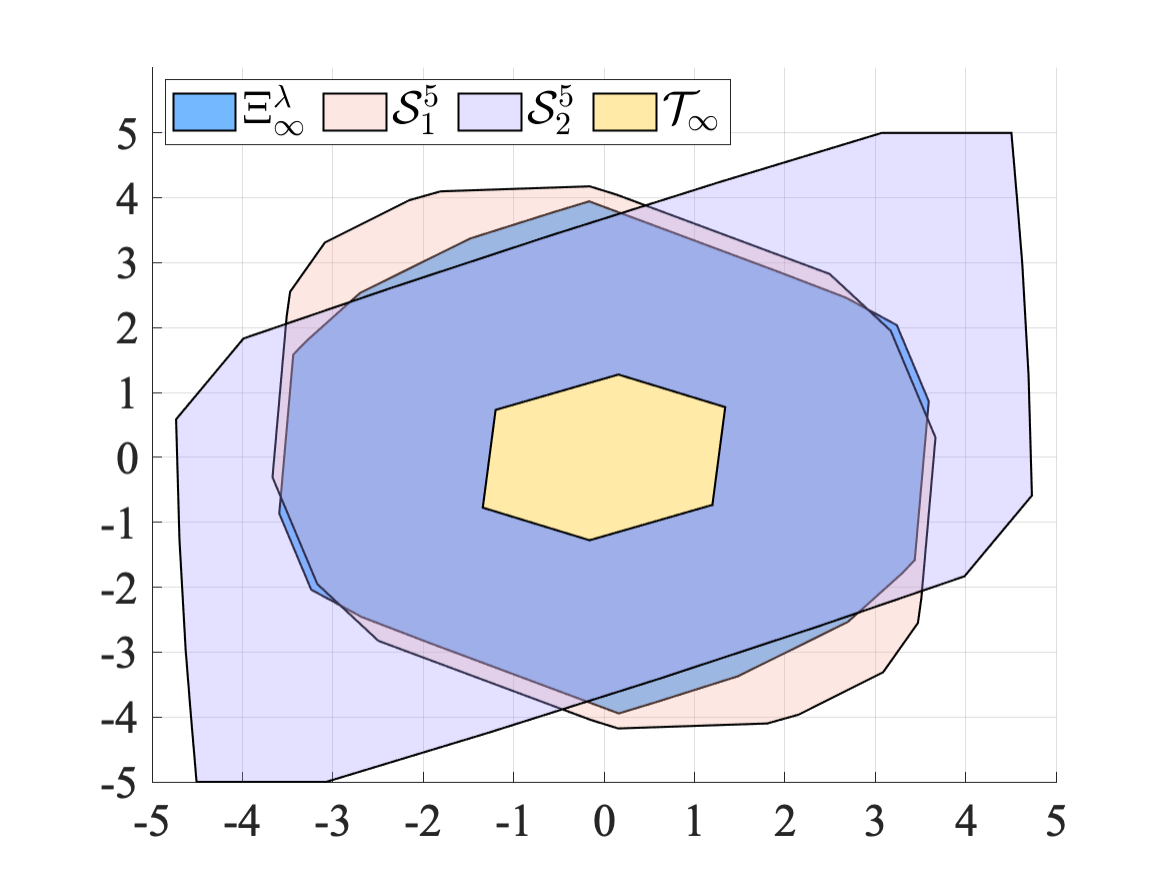}}\vspace{-0.3cm}
	\caption{The feasible regions for the standard MPC with $N_1^{\max} = N_2^{\max} = 9$ and $N_1^{\min} = N_2^{\min} = 7$ and the corresponding switching feasible set}\vspace{-0.3cm}
	\label{fig:set2}
\end{figure}

\textit{$\bullet$ The relationship between the set $\Xi_{\infty}^\lambda$ and the set $\mathcal{S}_m^n$:}
The set $\mathcal{S}_{m}^{n} = \mathcal{P}_m^n(\mathcal{T}_{\infty})$ defines the $n$-step domain of attraction for the $m$-th subsystem under MPC constraints \eqref{MPC_c1}, \eqref{MPC_c2}, and \eqref{ter}. As discussed in Remark \ref{remark_feasible_sw}, the switching feasible set $\Xi_{\infty}^{\lambda}$ is a subset of the maximal-horizon intersection $\bigcap_{m\in \mathcal{M}} \mathcal{S}_m^{N_m^{\max}}$. Furthermore, there is no explicit inclusion relationship between $\Xi_{\infty}^{\lambda}$ and the minimal-horizon intersection $\bigcap_{m\in \mathcal{M}} \mathcal{S}_m^{N_m^{\min}}$.

These geometric relationships are visually demonstrated by Figs. \ref{fig:set} and \ref{fig:set2}. Fig. \ref{fig:set} plots $\Xi_{\infty}^{\lambda}$ together with the sets $\mathcal{S}_{m}^{n}$ for all $m \in \mathcal{M}$ and $n \in \{N_m^{\min}, N_m^{\max}\}$ (set to $5$ and $7$, respectively). The set $\Xi_{\infty}^{\lambda}$ is strictly contained in the maximal-horizon intersection $\mathcal{S}_{1}^{7} \cap \mathcal{S}_{2}^{7}$, whereas the minimal-horizon inclusion $\mathcal{S}_{1}^{5} \cap \mathcal{S}_{2}^{5} \subseteq \Xi_{\infty}^{\lambda}$ is not guaranteed. Fig. \ref{fig:set2} evaluates the sets with extended prediction horizons of $N_m^{\max} = 9$ and $N_m^{\min} = 7$. As shown, neither $\mathcal{S}_{1}^{7} \cap \mathcal{S}_{2}^{7} \subseteq \Xi_{\infty}^{\lambda}$ nor $\Xi_{\infty}^{\lambda} \subseteq \mathcal{S}_{1}^{7} \cap \mathcal{S}_{2}^{7}$ holds.

\begin{figure}
	\centering
	\includegraphics[width=0.65\linewidth]{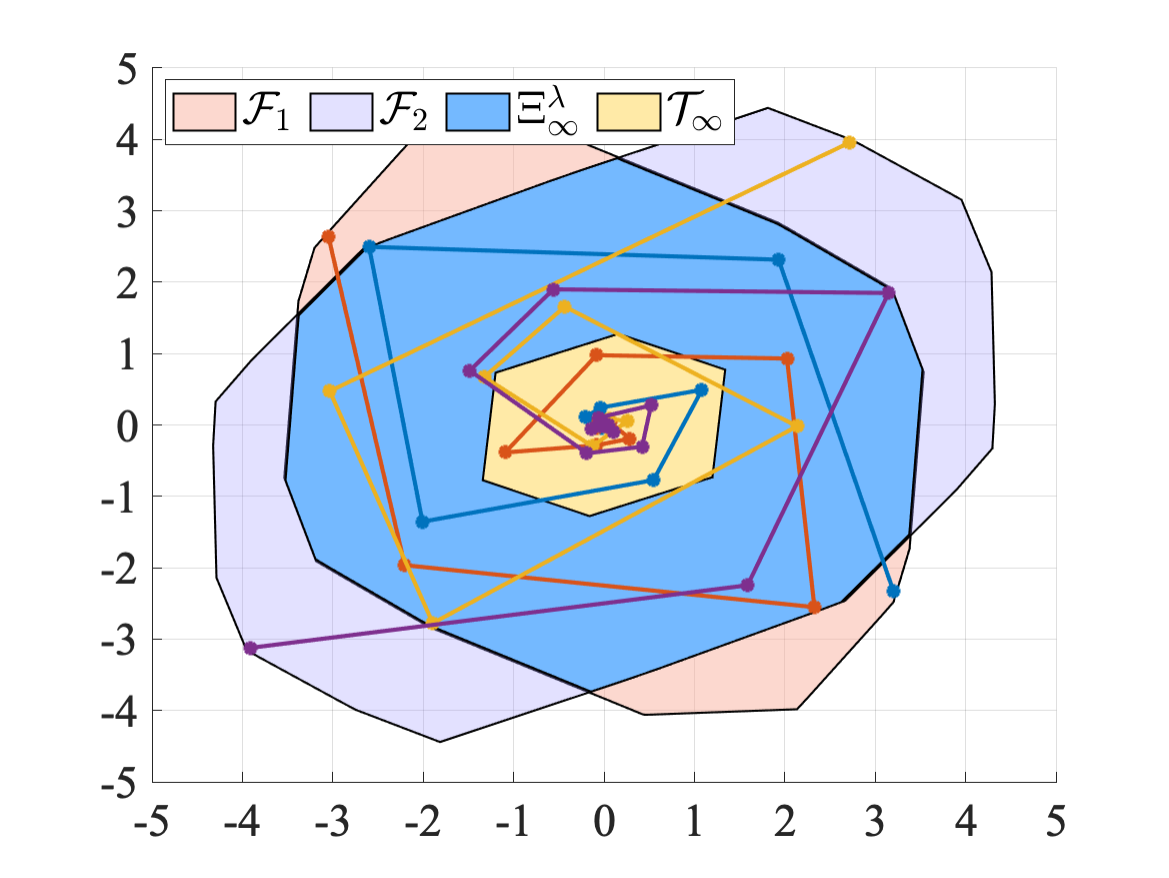}\vspace{-0.3cm}
	\caption{The domains  of attraction and the state trajectories from four different initial points}
	\label{fig:stateset}
\end{figure}
\begin{figure}
	\centering
	\begin{minipage}[b]{0.5\textwidth} 
		\centering
		\subfigure[State trajectories	\label{fig:state}]{\includegraphics[width=0.85\linewidth]{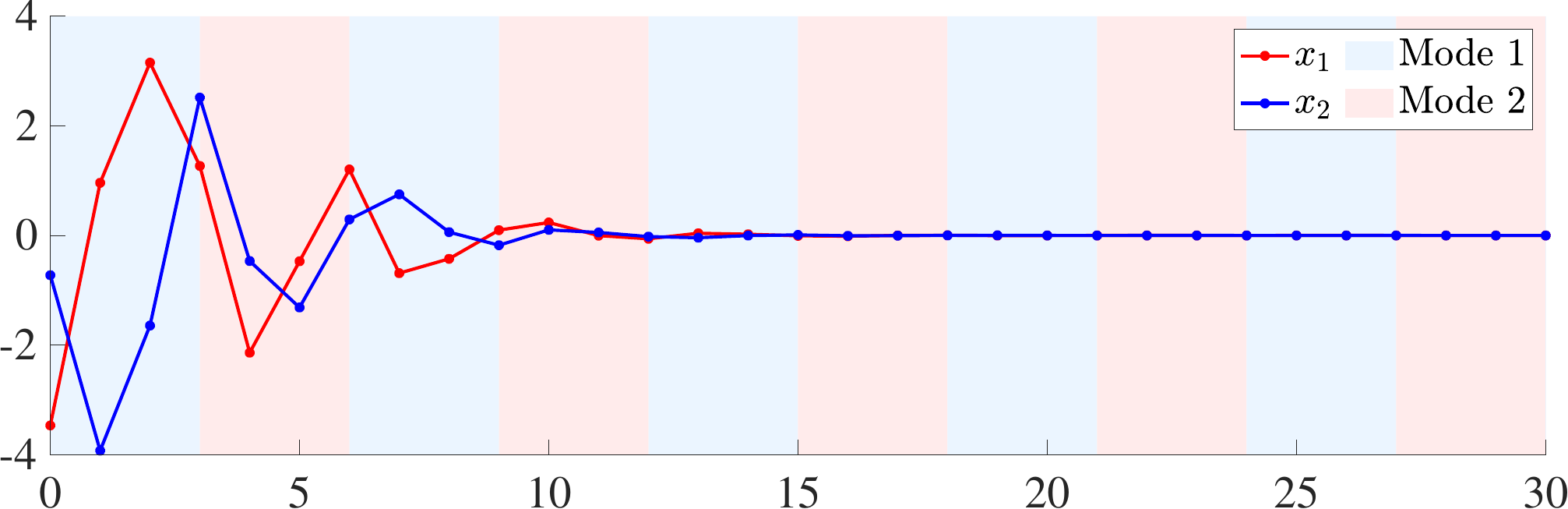}}
		\subfigure[Control input	\label{fig:switching}]{\includegraphics[width=0.85\linewidth]{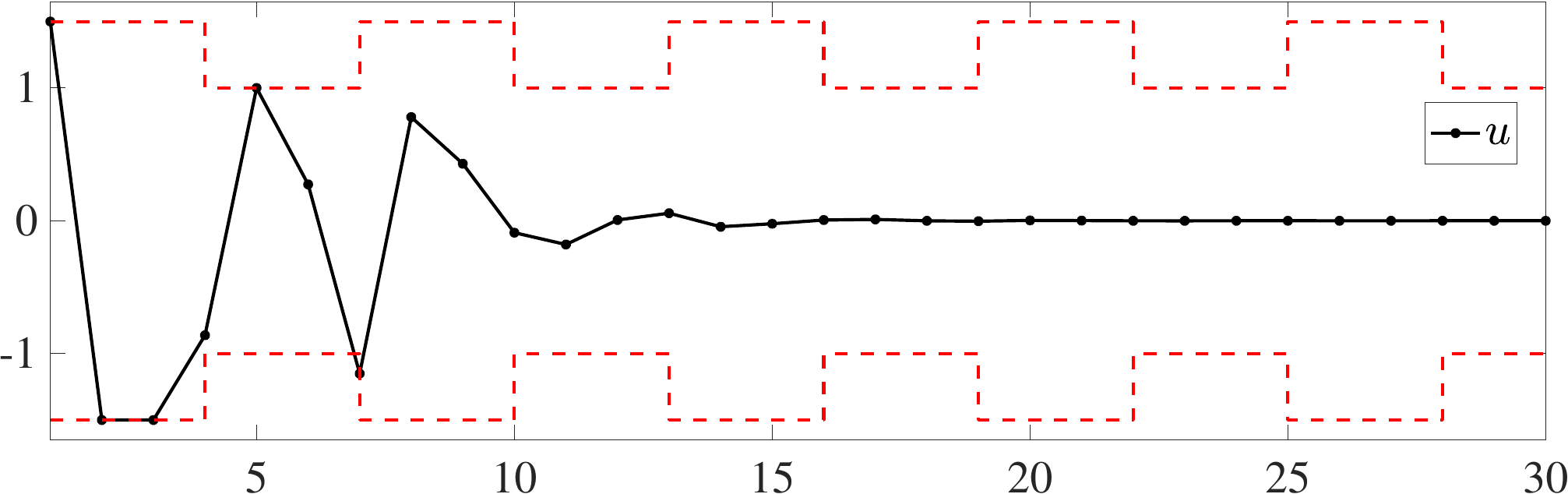}}\vspace{-0.3cm}
		\caption{State and input trajectories with the system initialized in subsystem 1}\label{fig:switching_state}
	\end{minipage}
\end{figure}\begin{figure}
	\begin{minipage}[b]{0.5\textwidth} 
		\centering
		\subfigure[State trajectories	\label{fig:state1}]{\includegraphics[width=0.85\linewidth]{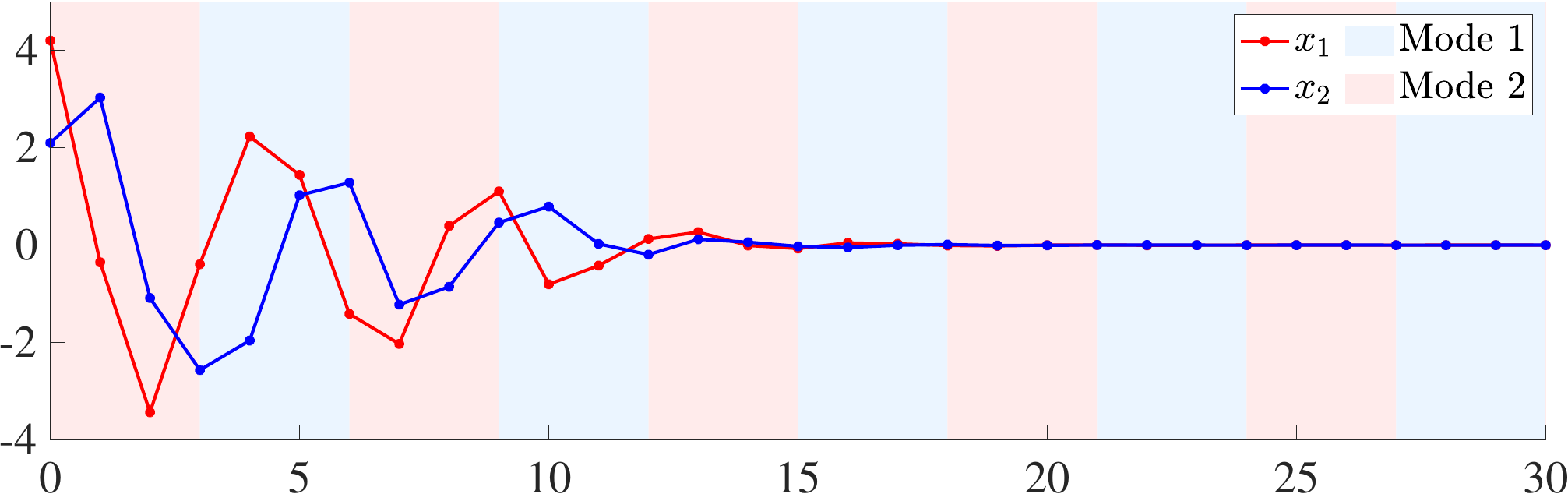}}
	\subfigure[Control input	\label{fig:switching1}]{\includegraphics[width=0.85\linewidth]{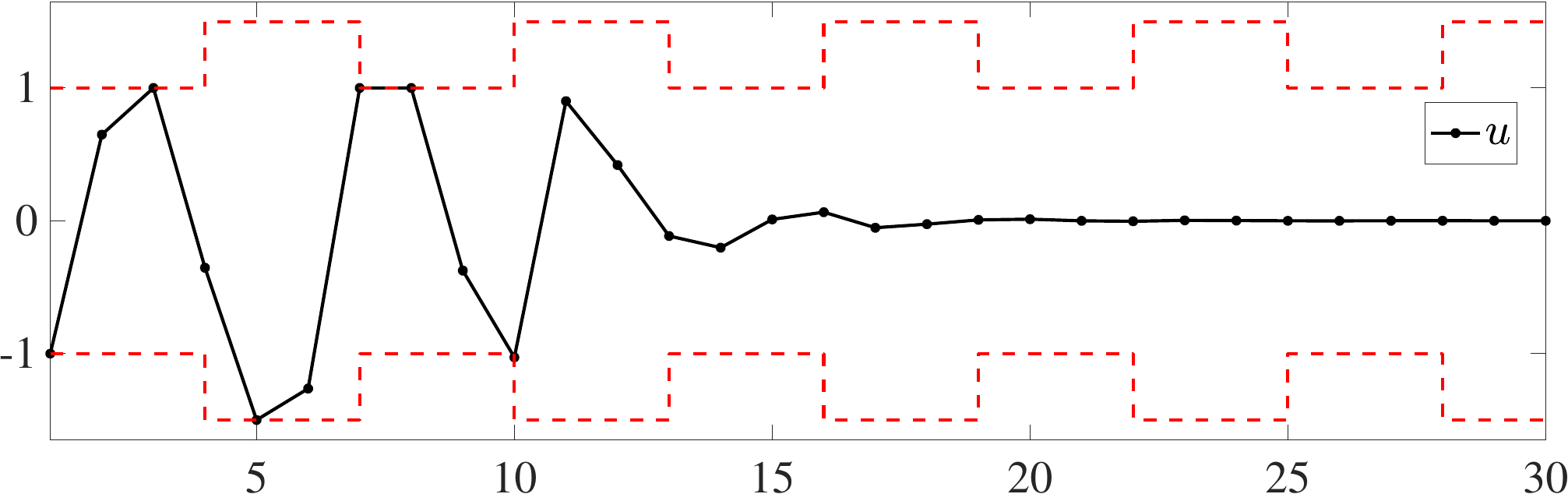}}\vspace{-0.3cm}
		\caption{State and input trajectories with the system initialized in subsystem 2}\label{fig:switching_state1}
	\end{minipage}
\end{figure}

\textit{$\bullet$ The effectiveness of the \textit{variable-horizon} switched MPC}: 
Fig. \ref{fig:stateset} depicts the domains of attraction derived from \textit{Algorithm \ref{result_alg_fs}}, superimposed with state trajectories starting from four distinct initial conditions. Under the mode-dependent dwell-time constraint, every initial state in the  domain of attraction is feasible for  the proposed switched MPC problem \eqref{MPC}.
Quantitatively, the switching feasible set has a volume of 38.5133 and is strictly contained within $\bigcap_{m \in \mathcal{M}} \mathcal{F}_m$, which has a volume of 38.7460.
\textcolor{black}{Figs. \ref{fig:switching_state} and \ref{fig:switching_state1} illustrate the state trajectories, control inputs, and the corresponding switching signals when the system is initialized in subsystem $1$ and subsystem $2$, respectively. As shown, all state trajectories successfully approach the origin over the $20$-step simulation.}

\textit{$\bullet$ Comparison with (\cite{OngMPC}}):
Fig. \ref{fig:campare} illustrates the domains of attraction obtained using the proposed method and the method in \cite{OngMPC}. $\mathcal{F}_m$ and $\mathcal{F}_{m}^{2016}$ represent the domains of attraction for the proposed method and the scheme in \cite{OngMPC}, respectively.
The key observations are as follows: 1) $\mathcal{F}_{m}{2016}$ is a strict subset of $\mathcal{F}_m$, meaning that every state admissible under the method in \cite{OngMPC} is also feasible for the proposed switched MPC (whereas the converse is false), which naturally results in a strictly larger volume for $\mathcal{F}_m$; 2) The proposed MPC scheme is computationally more efficient as it bypasses the max-min optimization required in \cite{OngMPC}.

\begin{figure}
	\centering
 	\includegraphics[width=0.65\linewidth]{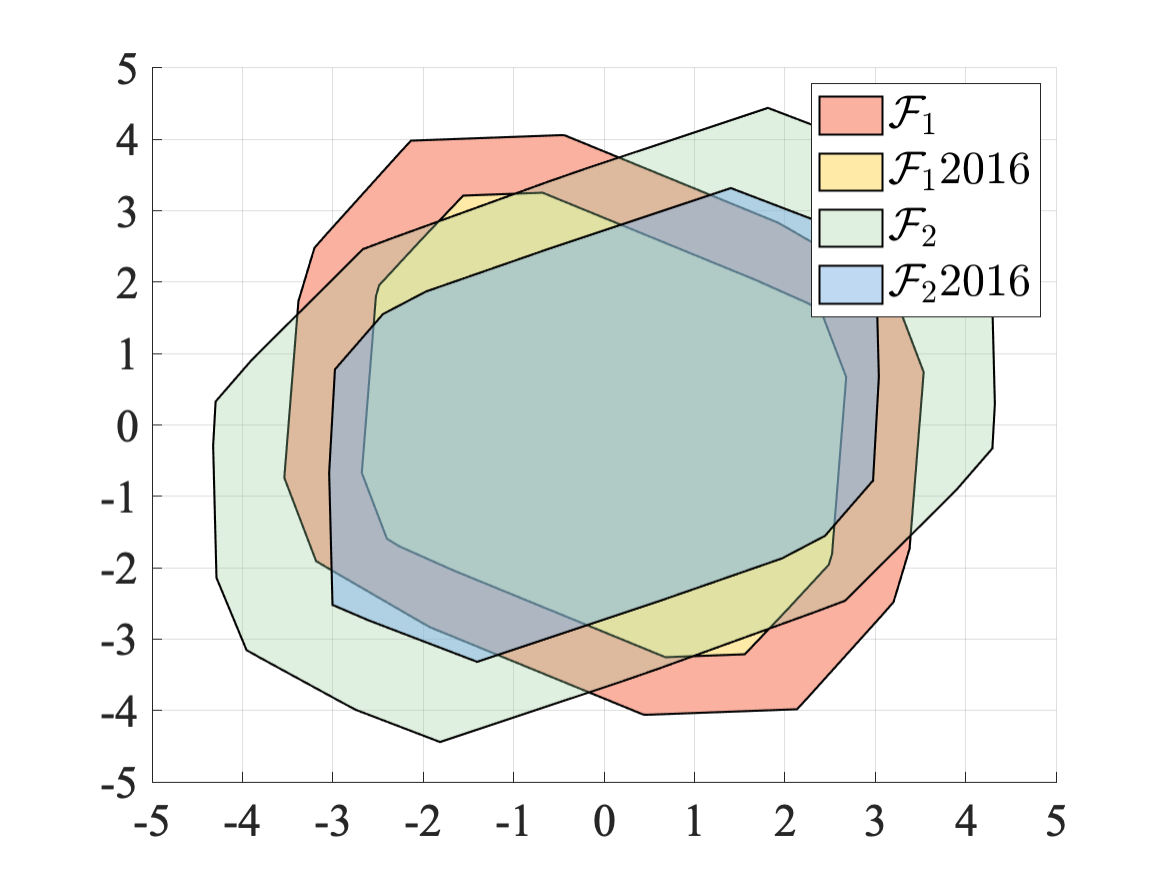}\vspace{-0.3cm}
	\caption{The domains of attraction using the MPC proposed in this paper and \cite{OngMPC} for switched system in this paper}\vspace{-0.2cm}
	\label{fig:campare}
\end{figure}
%

\begin{table}[b]
	\centering
	\caption{Comparison of dwell-time constraints using the method in \cite{zhang_switched_2016} and the proposed scheme}
	\label{table}  
	\renewcommand{\arraystretch}{1.1}  
	\begin{tabular}{l *{9}{c}}  
		\toprule
		\textbf{DT} & \multicolumn{8}{c}{Zhang et al. (2016)} & \textbf{Proposed} \\
		\cmidrule(lr){2-9} \cmidrule(lr){10-10}
		Stage & 1 & 2 & 3 & 4 & 5 & 6&7&$\ge 8$ & $\ge 1$ \\
		\midrule
		$\tau_{d,1}$ & 5 & 5 & 5 & 4 & 4 & 3&4&2 & 2 \\
		$\tau_{d,2}$ & 4 & 4 & 4 &2 & 2 & 2 &2&2 & 2 \\
		\bottomrule
	\end{tabular}
\end{table}

\textit{$\bullet$ Comparison with (\cite{zhang_switched_2016}}): 
The dwell-time constraints adopted here exhibit greater flexibility than \textcolor{black}{those} in \cite{zhang_switched_2016}. Table \ref{table} compares the dwell-time constraints required by the method in \cite{zhang_switched_2016} and the proposed approach. 
The constraints in \cite{zhang_switched_2016} are derived through complex set operations.  These operations increase the offline computational burden, while the requirement to ensure feasibility across all admissible mode transitions may lead to conservative dwell-time bounds.  Moreover, their calculation procedure is computationally demanding, as the dwell-time bounds are strongly coupled with the number of switches.
In contrast, the proposed MPC scheme yields less restrictive dwell-time constraints. 
Furthermore, when the system matrices from (\cite{zhang_switched_2016}) are evaluated within our framework, the required dwell time is significantly reduced to $\tau_{d,1}=2$ and $\tau_{d,2}=4$. Notably, these shortened durations would be completely infeasible under the stringent conditions established in \cite{zhang_switched_2016}, clearly demonstrating the expanded operational flexibility of our approach.

\begin{figure}
	\centering
	\subfigure[$\tau_1=\tau_2 = 2$]{\includegraphics[width=0.65\linewidth]{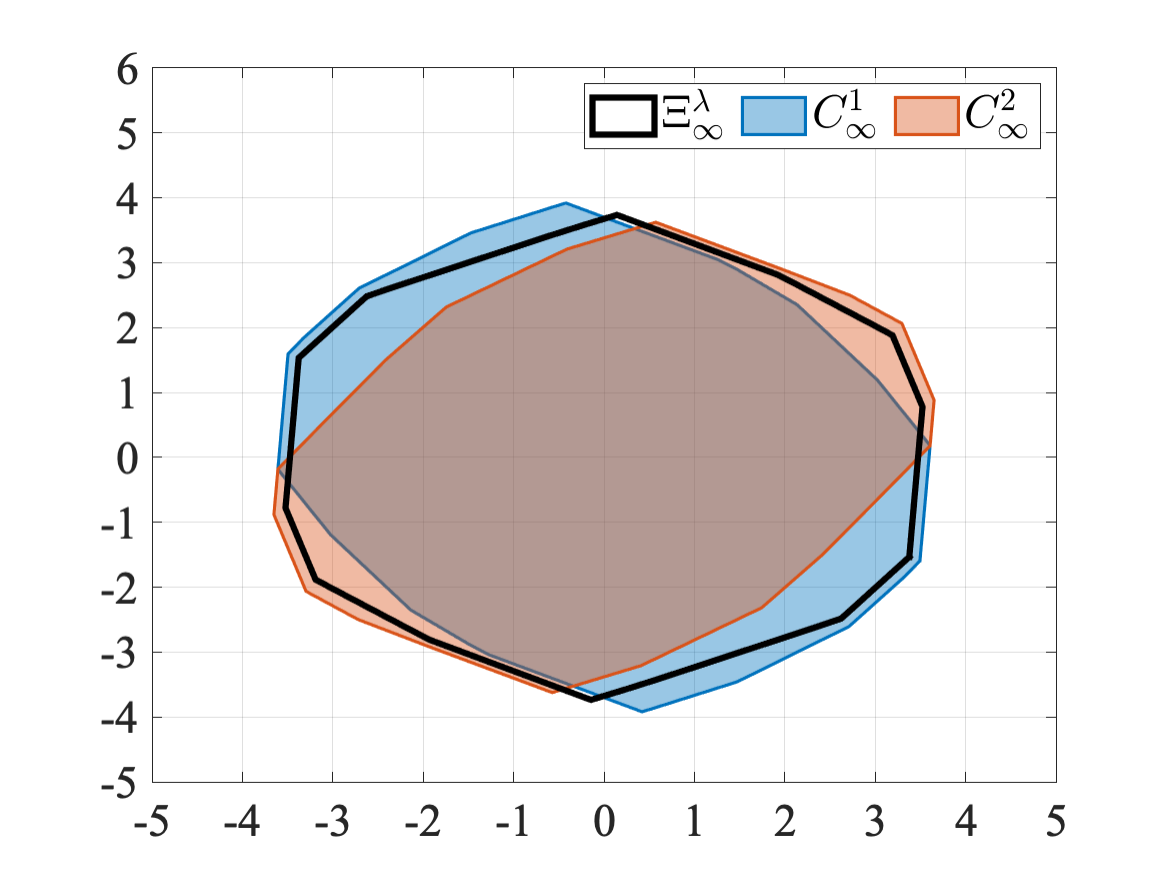}}\vspace{-0.2cm}
	\subfigure[$\tau_1 = 4,$ $\tau_2 = 5$]{\includegraphics[width=0.65\linewidth]{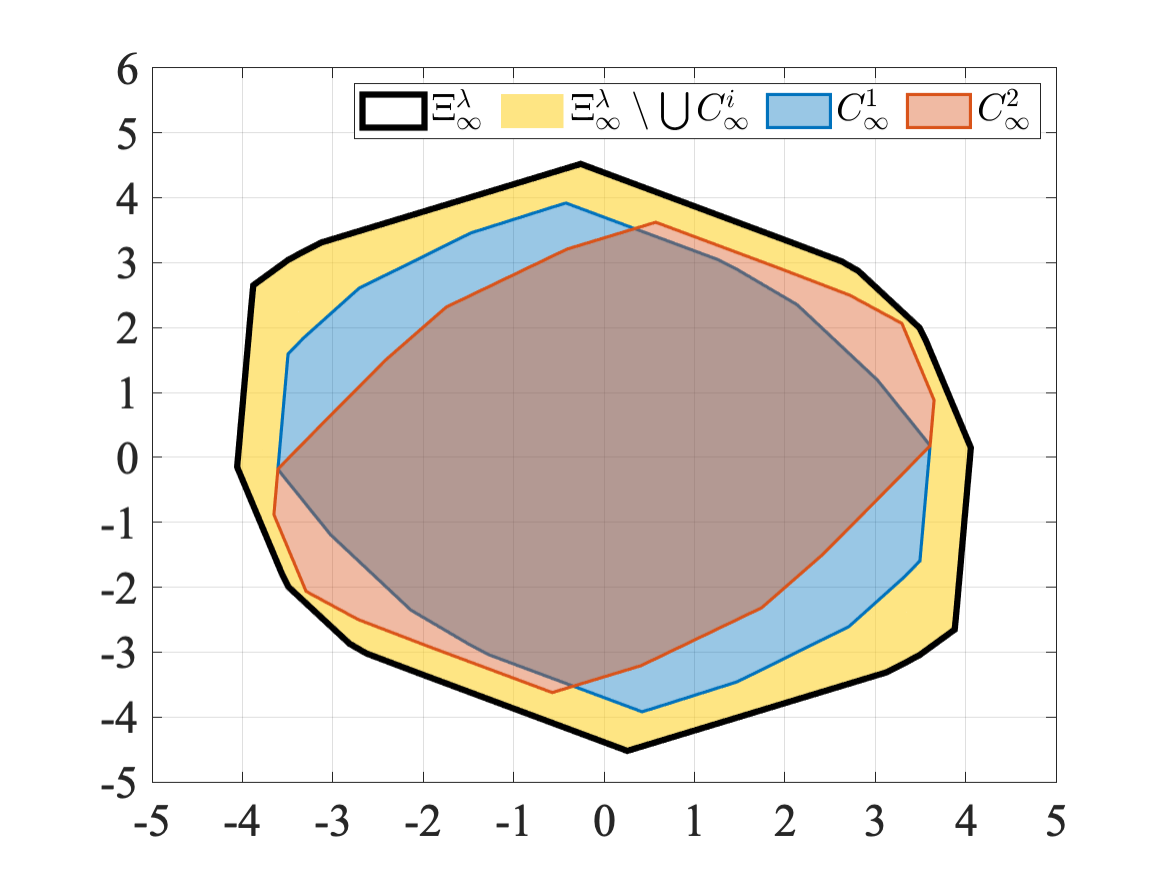}}\vspace{-0.3cm}
	\caption{\textcolor{black}{The switching feasible set $\Xi_{\infty}^{\lambda}$ of Algorithm 1 and the switch-RCI sets $C^i_{\infty}$ of Algorithm 1 in \cite{Danielson2019_switchRCI}}}\vspace{-0.2cm}
	\label{fig:main}
\end{figure}

{\color{black}
	\textit{$\bullet$ Comparison with (\cite{Danielson2019_switchRCI}}): 
	To visually assess the geometric properties of the feasible regions, Fig. \ref{fig:main} compares our switching feasible set $\Xi_{\infty}^{\lambda}$ with the exact switch-RCI sets $C_{\infty}^i$ computed by the method in \cite{Danielson2019_switchRCI}. As discussed in Remark \ref{remark_danielson_diff}, these two methods impose different set conditions, and    neither definition  implies a general inclusion relationship. 
	The switch-RCI construction enforces one-step control invariance, whereas the proposed method allows the state to temporarily leave \(\Xi_{\infty}^{\lambda}\) provided that  the physical constraints  remain satisfied and  the prescribed multi-step transition is completed. 
	The effect of dwell time can be observed by comparing the two parameter settings.  As seen in Fig. \ref{fig:main}(a), under short dwell-time constraints ($\tau_1=\tau_2=2$), the feasible regions of both methods are comparable. However, when the dwell time increases to $\tau_1=4$ and $\tau_2=5$ in Fig. \ref{fig:main}(b), the proposed switching feasible set expands, whereas the switch-RCI sets remain constrained by their one-step invariance conditions. In this example, the proposed algorithm uses the additional dwell-time steps to retain additional states in \(\Xi_\infty^\lambda\).}

\color{black}
\textit{$\bullet$ Comparative evaluation and computational complexity analysis}:
To evaluate the conservatism and computational efficiency of the proposed scheme, comprehensive comparisons were conducted against existing approaches (\cite{OngMPC, zhang_switched_2016, Danielson2019_switchRCI}), see Table \ref{tab:computation_time}.   Additional system parameters are provided in the Appendix. To ensure fairness, the prediction horizon for the reference methods was uniformly set to $N=5$. The reported offline computation times are averages over $10$ independent runs for method--scenario combinations whose set computations completed; an incomplete computation is reported as NC instead. For online MPC optimization, initial states were randomly sampled from the intersection of the initial feasible sets of all evaluated methods. For each $20$-step run, all applicable methods shared the sampled initial state and a switching signal derived from the most conservative constraints in \cite{zhang_switched_2016}; online time was averaged per step across runs.



Regarding  computational efficiency, the proposed offline phase requires slightly more time than that of  \cite{OngMPC}, owing to the explicit calculation of the switching feasible set. However, it exhibits a significant speed advantage over the methods in \cite{zhang_switched_2016} and \cite{Danielson2019_switchRCI}. This advantage becomes particularly prominent as the number of subsystems ($M$) or state dimensions ($n_x$) increases; notably, for $(M,n_x)=(2,3)$, the offline set computation of \cite{Danielson2019_switchRCI} had not converged when the 500-s time limit was reached; hence no completed offline time or online time is reported for this case. During the online phase, our computation time is marginally higher than the baseline in \cite{zhang_switched_2016}. This slight overhead is fundamentally attributed to the longer prediction horizon inherently supported by our method. If the horizon in \cite{zhang_switched_2016} is extended to match (e.g., $N=7$ for $M=2, n_x=2$), their online calculation time rises to $0.1709$ s, demonstrating  that our online efficiency remains highly competitive. Overall, the proposed method achieves a favorable offline-online computational trade-off in the tested configurations, reducing offline computation time relative to set-based methods while maintaining comparable online efficiency.

\begin{table}
	\centering
	\renewcommand{\arraystretch}{1.1} 
	\caption{Quantitative comparison of computation times}
	\label{tab:computation_time}
	\resizebox{\columnwidth}{!}{%
	\begin{tabular}{clrc} 
		\toprule
		\multicolumn{2}{c}{\textbf{Scenario / Method}} & \textbf{Offline(s)} & \textbf{Online(s)} \\
		\midrule
		\multirow{4}{*}{\makecell{$M = 2$; \\ $n_x = 2$}} 
		& Proposed & 7.185~~~~ & 0.1697\\
		& \cite{OngMPC} & {1.656}~~~~ & 0.3730 \\
		& \cite{zhang_switched_2016} & 29.051~~~~~& 0.1669 \\
		& \cite{Danielson2019_switchRCI} &    37.015~~~~ & 0.1662 \\
		\midrule
		\multirow{4}{*}{\makecell{$M = 4$; \\ $n_x = 2$}} 
		& Proposed & 8.133~~~~ & 0.1832 \\
		& \cite{OngMPC} & 3.068~~~~ & 1.6579 \\
		& \cite{zhang_switched_2016} & 34.434~~~~ & 0.1796 \\
		& \cite{Danielson2019_switchRCI} &  268.422~~~~ & 0.1789 \\
		\midrule
		\multirow{4}{*}{\makecell{$M = 2$; \\ $n_x = 3$}} 
		& Proposed &   250.150~~~~ & 0.1791 \\
		& \cite{OngMPC} & 80.614~~~~ & 0.3526 \\
		& \cite{zhang_switched_2016} & 496.464~~~~~& 0.1845 \\
		& \cite{Danielson2019_switchRCI} & NC (500 s) & N/A \\
		\bottomrule
	\end{tabular}
}
NC indicates that the offline set computation had not converged when the 500-s time limit was reached. N/A indicates that no online computation time is available for this case.
\end{table}

\color{black}

\textit{$\bullet$ The application of short-horizon switched MPC}: 
The domains of attraction reported in \cite{zhang_switched_2016}, configured with $N_{1}^{\max} = N_{2}^{\max} = 7$, are larger than those of the proposed standard variable-horizon MPC \eqref{MPC}. This conservative limitation is effectively overcome by the proposed \textit{short-horizon} switched MPC \eqref{SHMPC1}.
In Fig. \ref{fig:setstate}, $\hat{\mathcal{F}}_m$ and $\tilde{\mathcal{F}}_m^{\mathcal{N}_m^{*}}$ represent the domains of attraction obtained using the method in \cite{zhang_switched_2016} and the \textit{short-horizon} switched MPC, respectively.
The remaining parameters are specified as $N^{\min}_1=N^{\min}_2=5,~N^{\max}_1=N^{\max}_2=7,~ \mathcal{N}_1^* = 11,~ \text{and~} \mathcal{N}_2^*= 10.$ For the first switching instant, the dwell-time constraints $\hat{\tau}_{d,1}= 6$ and $\hat{\tau}_{d,2}= 5$ align with those in \cite{zhang_switched_2016}. 
Notably,  the \textit{short-horizon} switched MPC delivers strictly larger domains of attraction than those reported in \cite{zhang_switched_2016}. This demonstrates the enhanced capability of the proposed approach to enlarge the feasible region.
Finally, Fig. \ref{fig:switching_state1short} plots the state trajectories under the \textit{short-horizon} switched MPC \eqref{SHMPC1}, confirming the closed-loop asymptotic stability of the system.

\begin{figure}[htbp]
\centering
\includegraphics[width=0.6\linewidth]{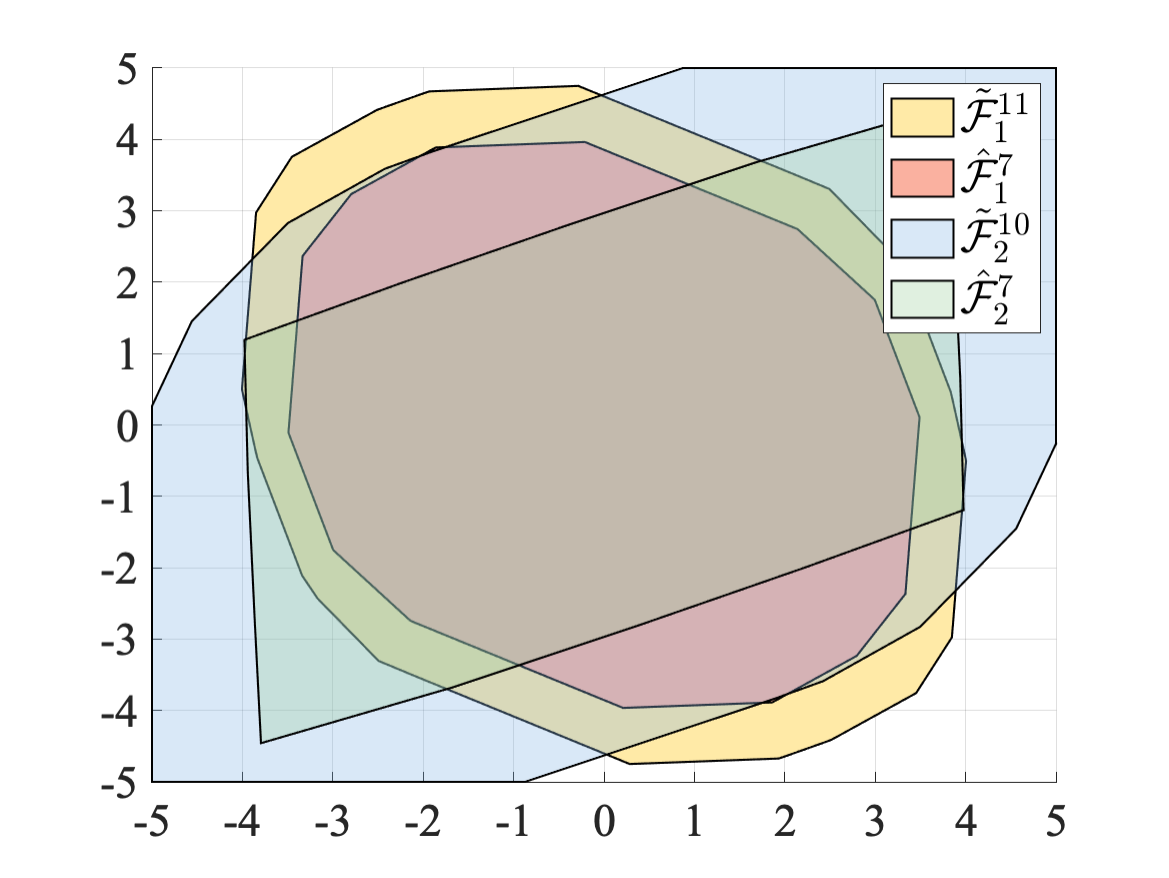}
\caption{The domains of attraction using $N^{\min}_1=N^{\min}_2=5$, $\mathcal{N}_1^* = 11$ and $\mathcal{N}_2^*= 10$. }\vspace{-0.2cm}
\label{fig:setstate}
\end{figure}

\begin{figure}[htbp]
\centering
\subfigure[State trajectories	\label{fig:state2short}]{\includegraphics[width=0.85\linewidth]{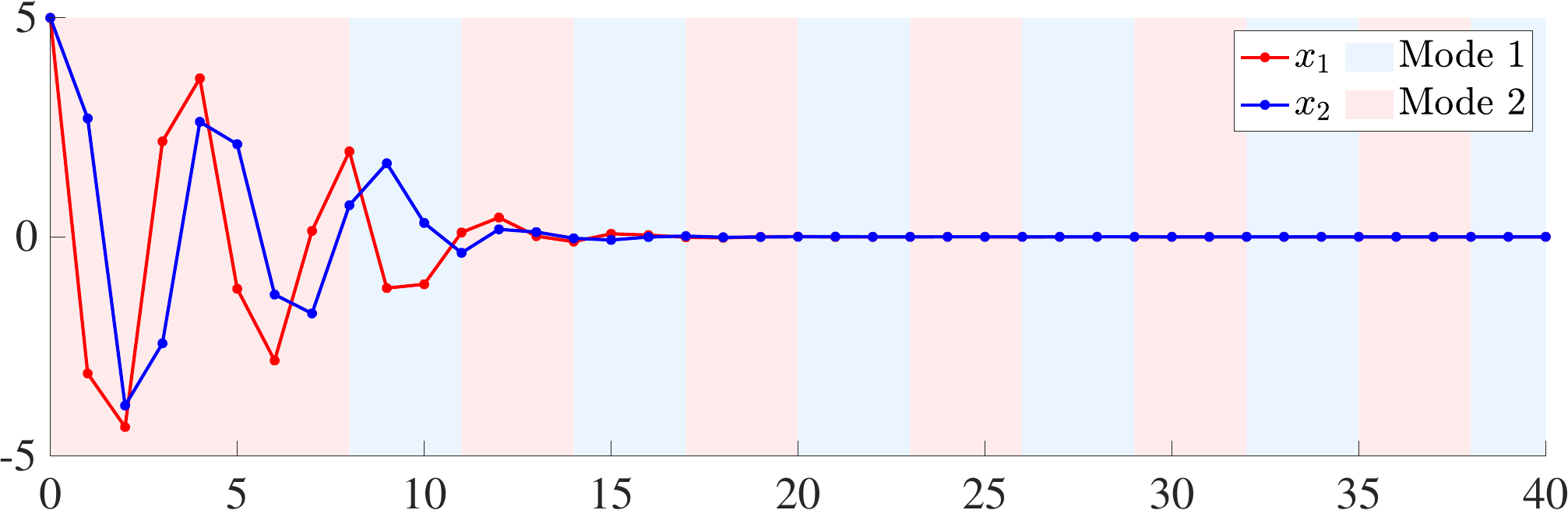}}
\subfigure[Control input	\label{fig:switching2short}]{\includegraphics[width=0.85\linewidth]{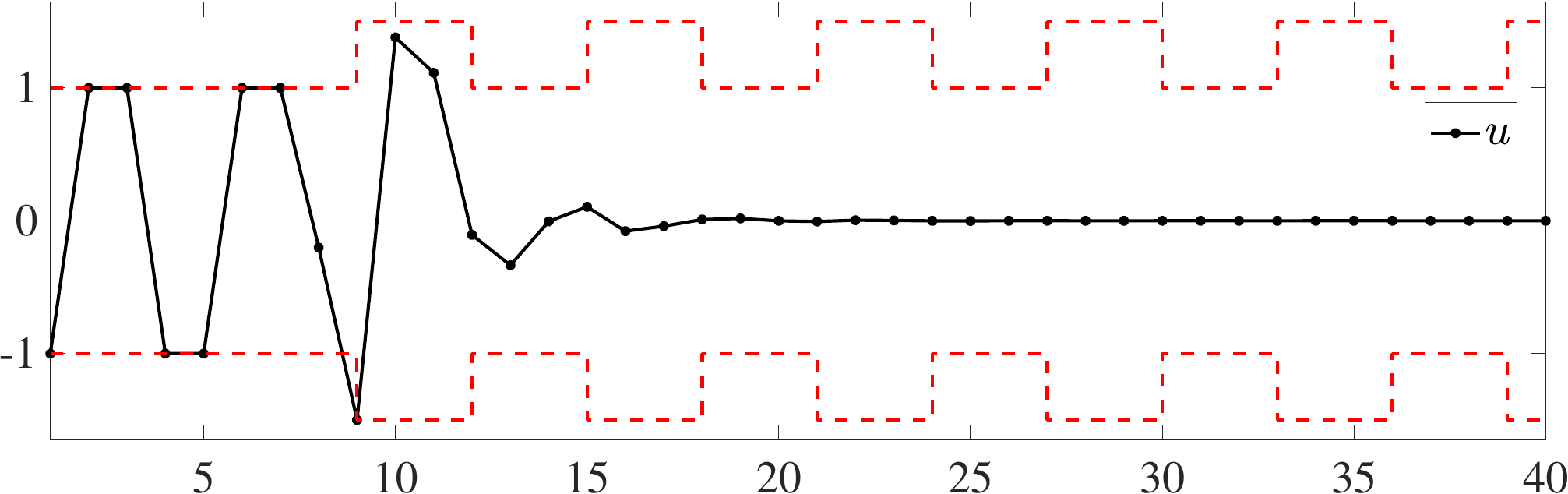}} \vspace{-0.3cm}
\caption{State and input trajectories under the short-horizon switched MPC}\label{fig:switching_state1short} \vspace{-0.2cm}
	\end{figure}

%
%

	\section{Conclusion} \label{sec:col}

	This paper presents a \textit{variable-horizon} switched MPC scheme for switched systems. An algorithm is developed to compute the switching feasible set. By leveraging this set, persistent feasibility is guaranteed, provided that the conditions regarding  the prediction horizon and dwell-time constraints are satisfied. The dwell-time conditions are derived from the corresponding unconstrained closed-loop subsystems, which also ensure the stability of  the switched systems under the proposed MPC. Furthermore, two \textit{short-horizon} switched MPC schemes are proposed. They can relax the constraint that dwell-time is less than the maximum prediction horizon, thereby expanding the domain of attraction. Simulations with comparisons \textcolor{black}{illustrate the performance} of the proposed MPC scheme for switched systems.



\bibliographystyle{apalike}   
\bibliography{mybibfile}           


\color{black}
\appendix
\section{System Parameters for Simulations}
\label{app:parameters}

\subsection{2D Switched System with four modes}
To evaluate the proposed algorithm on a system with more modes, two additional subsystems (Modes 3 and 4) are appended to the 2D system, given by 
\begin{align*}
	A_3 &= \begin{bmatrix} -0.0724 & 0.9009 \\ 0.6191 & 0.5519 \end{bmatrix}, \quad
	B_3 = \begin{bmatrix} 0.2 \\ -0.5 \end{bmatrix}, \\
	A_4 &= \begin{bmatrix} 0.8678 & 0.3657 \\ 0.8176 & -0.3591 \end{bmatrix}, \quad
	B_4 = \begin{bmatrix} -0.3 \\ 0.5 \end{bmatrix}.
\end{align*}
The state and control constraints are uniform for both modes, given by $\mathbb{X}_3 = \mathbb{X}_4 = \{ x \in \mathbb{R}^2 : \|x\|_{\infty} \leq 5\}$ and $\mathbb{U}_3 = \mathbb{U}_4 = \{ u \in \mathbb{R} : |u| \leq 1\}$. The weighting matrices are set to $Q_3 = Q_4 = I$ and $R_3 = R_4 = 4$. The terminal penalty matrices $P_m$ and feedback gains $K_m$ ($m \in \{3,4\}$) are derived via unconstrained LQR. The required dwell time is $\tau_{d,3} = \tau_{d,4} = 2$. The variable prediction horizons are bounded by $N^{\min}_3 = N^{\min}_4 = 5$ and $N^{\max}_3 = N^{\max}_4 = 7$.


\begin{figure*}[b]
	\centering
	\includegraphics[width=\linewidth]{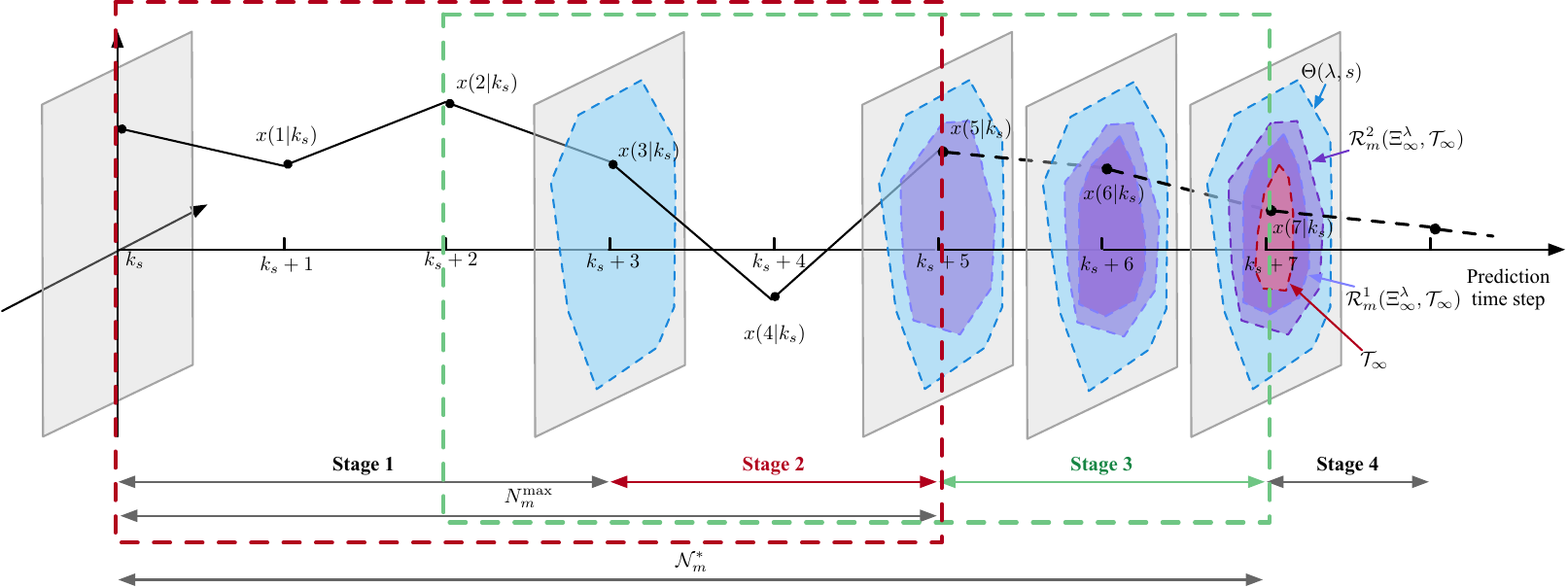}
	\caption{
		Multi-mode prediction paradigm of the short-horizon switched MPC with $\mathcal{N}_m^*=7$, $N_m^{\max}=5$, and $N_m^{\min}=2$.
		The red dashed box shows the prediction at the switching instant $k_s$, while the green dashed box (at $k_s+2$) coincides with the standard variable-horizon MPC in Fig.~\ref{fig:paradigm}.
		At $k_s$, the control sequence steers the predicted state into $\Theta(\lambda,s)$ from step $N_m^{\max}-N_m^{\min}$ onward, and forces the $N_m^{\max}$-step-ahead state into the predecessor set $\mathcal{R}_m^{\rho_k}(\Xi_\infty^\lambda,\mathcal{T}_\infty)$, ensuring arrival at $\mathcal{T}_\infty$ within $\mathcal{N}_m^*$ steps.
		The purple region realizes a transitional phase: although the prediction horizon is shorter than $\mathcal{N}_m^*$, the nested predecessor structure guarantees that the terminal requirement is met once the mode duration is sufficiently long.
		After entering $\mathcal{T}_\infty$, the predicted trajectory follows the local LQR law $u(k)=K_{\sigma(k)}x(k)$.
	}
	\label{fig:paradigmsh}
\end{figure*}
\subsection{3D Switched System}
For a higher-dimensional test, a 3D switched system comprising two subsystems is considered. The system dynamics are
\begin{equation*}
	A_1 = \begin{bmatrix} -0.8125 & 0.7000 & -0.9375 \\ -0.8125 & -0.2375 & 0.1250 \\ -1.2375 & 0.1250 & 0.9500 \end{bmatrix}, \quad
	B_1 = \begin{bmatrix} -1 \\ 1 \\ 1 \end{bmatrix},
\end{equation*}
\begin{equation*}
	A_2 = \begin{bmatrix} -0.7750 & 0.9375 & 0.0125 \\ 0.6125 & 1.1750 & -0.4375 \\ -0.6875 & 0.8000 & -0.4750 \end{bmatrix}, \quad
	B_2 = \begin{bmatrix} 1 \\ 1 \\ -0.5 \end{bmatrix}.
\end{equation*}
The state constraint sets are $\mathbb{X}_1 = \{ x \in \mathbb{R}^3 : \|x\|_{\infty} \leq 10\}$ and $\mathbb{X}_2 = \{ x \in \mathbb{R}^3 : \|x\|_{\infty} \leq 9\}$. The control inputs are bounded in $\mathbb{U}_1 = \{ u \in \mathbb{R} : \|u\| \leq 1.5\}$ and $\mathbb{U}_2 = \{ u \in \mathbb{R} : \|u\| \leq 2\}$. The LQR weighting matrices are specified as $Q_1 = 10I$, $Q_2 = 5I$, and $R_1 = R_2 = 1$, which similarly yield the corresponding $P_m$ and $K_m$. The horizon bounds are chosen as $N^{\min}_1 = N^{\min}_2 = 5$ and $N^{\max}_1 = N^{\max}_2 = 7$.

\section{Short-Horizon Switched MPC for {$\mathit{N^{\max}_m > \tau_{d,m}}$}}
A schematic diagram for the short-horizon switched MPC problem~\eqref{SHMPC1} is illustrated by Fig.~\ref{fig:paradigmsh}, with Table~\ref{table3} detailing the prediction horizon and state constraints at each step. Ideally, the state should enter the terminal set $\mathcal{T}_\infty$ within $\mathcal{N}_m^*$ steps. However, since the allowable maximum prediction horizon satisfies $N_m^{\max} < \mathcal{N}_m^*$, only a restricted prediction window (the red dashed box at $k_s$) is available at the switching instant. To bridge this temporal gap, intermediate transitional sets (the purple regions in Fig.~\ref{fig:paradigmsh}) are introduced as multi-step predecessor sets of $\mathcal{T}_\infty$.

The key distinction from the standard variable-horizon MPC (Fig.~\ref{fig:paradigm}) lies precisely in this transitional phase. Once the dwell time after the switch becomes sufficiently long (e.g., at $k_s+2$, represented by the green dashed box), the short-horizon scheme smoothly reduces to the standard variable-horizon MPC. A direct comparison between Tables~\ref{table1} and~\ref{table3} reveals that the terminal constraint in the short-horizon formulation is no longer $\mathcal{T}_\infty$ itself, but a recursively constructed predecessor set $\mathcal{R}_m^{\rho_k}(\Xi_\infty^\lambda,\mathcal{T}_\infty)$, which guarantees that the state reaches $\mathcal{T}_\infty$ exactly at step $\mathcal{N}_m^*$.

\begin{table*}[htbp]
	\centering
	\caption{Prediction horizon and state constraints at each instant and predicted step with $\mathcal{N}_m^* = 7$, $N^{\max}_m = 5$,  $N^{\min}_m = 2$ (see \eqref{SHMPC1})}
	\label{table3}  
	\renewcommand{\arraystretch}{1.2}  
	\begin{tabular}{c *{6}{c}}  
		\toprule
		& $k_s$ & $k_s+1$ & $k_s+2$ & $k_s+3$ & $k_s+4$ & $[k_s+5, k_{s+1})$ \\
		\midrule
		\textbf{Horizon} & 5 & 5 & 5 & 4 & 3 & 2 \\
		\textbf{Step 1} & $\mathbb{X}_m$ & $\mathbb{X}_m$ & $\Theta(\lambda,s)$ & $\Theta(\lambda,s)$ & $\Theta(\lambda,s)$ & $\Theta(\lambda,s)$ \\
		\textbf{Step 2} & $\mathbb{X}_m$ & $\Theta(\lambda,s)$ & $\Theta(\lambda,s)$ & $\Theta(\lambda,s)$ & $\Theta(\lambda,s)$ & $\mathcal{T}_\infty$ \\
		\textbf{Step 3} & $\Theta(\lambda,s)$ & $\Theta(\lambda,s)$ & $\Theta(\lambda,s)$ & $\Theta(\lambda,s)$ & $\mathcal{T}_\infty$ & N/A \\
		\textbf{Step 4} & $\Theta(\lambda,s)$ & $\Theta(\lambda,s)$ & $\Theta(\lambda,s)$ & $\mathcal{T}_\infty$ & N/A & N/A \\
		\textbf{Step 5} & $\mathcal{R}^2_m(\Xi_{\infty}^{\lambda},\mathcal{T}_\infty)$ & $\mathcal{R}^1_m(\Xi_{\infty}^{\lambda},\mathcal{T}_\infty)$ & $\mathcal{T}_\infty$ & N/A & N/A & N/A \\
		\bottomrule
	\end{tabular}
\end{table*}

\begin{table*}
	\centering
	\caption{Prediction horizon and state constraints at each instant and predicted step, with $\mathcal{N}_m^* = 8$, $\underline{\mathcal{N}}_m^* = 3$, $N^{\max}_m = 4$, and $N^{\min}_m = 2$ (see \eqref{SHMPC2})}
	\label{table4}
	\renewcommand{\arraystretch}{1.2} 
	\begin{tabular}{c *{7}{c}}  
		\toprule
		& $k_s$ & $k_s+1$ & $k_s+2$ & $k_s+3$ & $k_s+4$ & $k_s+5$ & $[k_s+6, k_{s+1})$ \\
		\midrule
		\textbf{Horizon} & 4 & 4 & 4 & 4 & 4 & 3 & 2 \\
		\textbf{Step 1} & $\mathbb{X}_m$ & $\mathbb{X}_m$ & $\mathbb{X}_m$ & $\mathbb{X}_m$ & $\Theta(\lambda,s)$ & $\Theta(\lambda,s)$ & $\Theta(\lambda,s)$ \\
		\textbf{Step 2} & $\mathbb{X}_m$ & $\mathbb{X}_m$ & $\mathbb{X}_m$ & $\Theta(\lambda,s)$ & $\Theta(\lambda,s)$ & $\Theta(\lambda,s)$ & $\mathcal{T}_\infty$ \\
		\textbf{Step 3} & $\mathbb{X}_m$ & $\mathbb{X}_m$ & $\Theta(\lambda,s)$ & $\Theta(\lambda,s)$ & $\Theta(\lambda,s)$ & $\mathcal{T}_\infty$ & N/A \\
		\textbf{Step 4} & ${\mathcal{R}}^1_m(\mathbb{X}_m, \mathcal{R}_m^3(\Xi_{\infty}^{\lambda},\mathcal{T}_{\infty}))$ & $\mathcal{R}_m^3(\Xi_\infty^\lambda, \mathcal{T}_\infty)$ & $\mathcal{R}_m^2(\Xi_\infty^\lambda, \mathcal{T}_\infty)$ & $\mathcal{R}_m^1(\Xi_\infty^\lambda, \mathcal{T}_\infty)$ & $\mathcal{T}_\infty$ & N/A & N/A \\
		\bottomrule
	\end{tabular}
\end{table*}

\begin{pf}
	\textbf{Persistent Feasibility:} 
	During the interval $[k_0, k_0 + \mathcal{N}_{m} - N^{\max}_{m})$, one has $\mathcal{D}(k) < \mathcal{N}_{m} - N^{\max}_{m}$ and \eqref{SHMPC1} reduces to a fixed-horizon MPC with horizon length $N_{m}^{\max}$ where $m = \sigma(k_0)$. 
	By the definition, $\mathcal{R}_m^{\delta_k}(\Xi_{\infty}^{\lambda},\mathcal{T}_{\infty})$ ensures that an admissible control drives the predicted state into $\mathcal{R}_m^{\delta_{k+1}}(\Xi_{\infty}^{\lambda},\mathcal{T}_{\infty}) = \mathcal{R}_m^{\delta_k - 1}(\Xi_{\infty}^{\lambda},\mathcal{T}_{\infty})$. 
	Consequently, feasibility at instant $k$ guarantees the feasibility of the problem at instant $k+1$. Iterating this logic over the initial interval ensures recursive feasibility, while feasibility for subsequent steps follows directly from Theorem \ref{thm_fea}.

	\textbf{Stability:}	
	The feasibility of \eqref{SHMPC1} requires an input sequence $u_{i|k} \in \mathbb{U}_{m}$ $(i \in \mathbb{Z}_{[0, N_{m}^{\max} - 1]})$ such that 1) $x_{i|k} \in \mathbb{X}_{m}$ $(i \in \mathbb{Z}_{[1, \Delta \mathcal{N}^*_m- 1]})$; 2) $x_{N_{m}^{\max}|k} \in \mathcal{R}_m^{\delta_k}(\Xi_{\infty}^{\lambda},\mathcal{T}_{\infty})$.  
	By the recursive nature of the predecessor operator and the construction of Algorithm \ref{result_alg_fs}, the set of all initial states $x_{0|k}$ satisfying these exact multi-step constraints is precisely the domain of attraction $\mathcal{F}_m$. 
	Thus, problem \eqref{SHMPC1} is mathematically equivalent to the standard switched MPC problem \eqref{MPC} configured with maximum horizon $\mathcal{N}^*_m$ and minimum horizon $N^{\min}_m$. By Theorem \ref{result_thm_sta}, asymptotic stability is guaranteed. \qed
\end{pf}

\section{Short-Horizon Switched MPC for {$\mathit{N^{\max}_m \le \tau_{d,m}}$}}
Table~\ref{tab:dwell_time_comparison} compares the required minimum dwell-time constraints obtained from the method in Zhang et al. (2016) and the proposed short-horizon switched MPC scheme for a switched linear system with four subsystems.
A further comparison with the two-subsystem case in Table~\ref{table} reveals a notable limitation of the method of Zhang et al. (2016).
Even when the first two subsystem matrices $(A_1,B_1)$ and $(A_2,B_2)$ remain identical, the dwell-time constraints computed by Zhang et al. (2016) change significantly simply due to the increase in the total number of subsystems.
More critically, these variations exhibit no clear monotonicity or physical interpretability.
For example, in the two-subsystem scenario, $\tau_{d,1}$ eventually converges to $2$ after stage~8,
whereas in the four-subsystem case, $\tau_{d,1}$ fluctuates between $2$ and $7$ across different stages despite the unchanged subsystem dynamics.
This phenomenon indicates that the conservatism of Zhang et al. (2016) originates not from the intrinsic properties of individual subsystems, but from the combinatorial complexity of multi-step reachable-set operations across all modes.

In sharp contrast, the proposed scheme yields a consistent dwell-time requirement $\tau_{d,i}=2$ for all subsystems, regardless of the total number of modes or the prediction stage. This stage-invariant and subsystem-independent result demonstrates that the proposed method successfully decouples the dwell-time condition from the combinatorial explosion inherent in switched systems, thereby offering significantly improved scalability and switching flexibility.

\begin{table}
	\centering
	\caption{Comparison of dwell-time constraints using the method in Zhang et al. (2016) and the proposed scheme}
	\label{tab:dwell_time_comparison}
	\renewcommand{\arraystretch}{1}  
	\begin{tabular}{lcccccccc}
		\toprule
		\textbf{DT} & \multicolumn{7}{c}{Zhang et al. (2016)} & \textbf{Proposed} \\
		\cmidrule(lr){2-8} \cmidrule(lr){9-9}
		Stage & 1 & 2 & 3 & 4 & 5 & 6& $\geq 7$& $\ge 1$ \\
		\midrule
		$\tau_{d,1}$ & 7 & 3  & 3  &  4 &  4 & 4  & 2 & 2 \\
		$\tau_{d,2}$ & 7  & 6  &  6 &  4 &2  & 4  &2&2  \\
		$\tau_{d,3}$ &  6 &  6 &   5&  4 &  4 & 4  &4 &2 \\
		$\tau_{d,4}$ & 2  &2   &  2 &2   &2   &  2 &  2&2 \\
		\bottomrule
	\end{tabular}
\end{table}


\end{document}